\documentclass{article}
\usepackage{kr}
\usepackage{times}
\usepackage{enumitem}
\usepackage{soul}
\usepackage{url}
\usepackage[hidelinks]{hyperref}
\usepackage[utf8]{inputenc}
\usepackage[small]{caption}
\usepackage[all]{xy}
\usepackage{graphicx}
\usepackage{mathtools}
\usepackage{amsmath}
\usepackage{amssymb}
\usepackage{amsthm}
\usepackage{booktabs}
\usepackage{algorithm}
\usepackage{algorithmic}
\usepackage{amsfonts}
\usepackage{xcolor}
\usepackage{multicol}
\usepackage{multirow}
\usepackage{stmaryrd}
\usepackage{thm-restate} 
\usepackage{tikz}

\newcommand{\Var}{\mathsf{Var}}

\newcommand{\Term}{\mathsf{Term}}

\newcommand{\At}{\mathsf{At}}

\newcommand{\CN}{\mathsf{CN}}
\newcommand{\IN}{\mathsf{IN}}
\newcommand{\RN}{\mathsf{RN}}
\newcommand{\RNinv}{\mathsf{RN}^{\pm}}

\newcommand{\ALCHI}{\mathcal{ALCHI}}
\newcommand{\Rep}{\mathtt{Rep}}
\newcommand{\close}{\mathbb{C}}
\newcommand{\closeRep}{\mathtt{CRep}}

\newcommand{\four}{\mathbf{4}}
\newcommand{\true}{\mathbf{T}}
\newcommand{\both}{\mathbf{B}}
\newcommand{\neither}{\mathbf{N}}
\newcommand{\false}{\mathbf{F}}

\newcommand{\cl}{\mathsf{cl}}
\newcommand{\ctr}{\mathsf{ctr}}

\newcommand{\nnf}{\mathsf{NNF}}
\newcommand{\ELHInegtrianglefour}{{\mathcal{ELHI}_\neg}^\four_\triangle}
\newcommand{\ALCHItrianglefour}{\mathcal{ALCHI}^\four_\triangle}

\newcommand{\ans}{\mathsf{ans}}

\newcommand{\brave}{\mathtt{brave}}
\newcommand{\IAR}{\mathtt{IAR}}
\newcommand{\AR}{\mathtt{AR}}

\newcommand{\CAR}{\mathtt{CAR}}

\newcommand{\mmsf}{\mathsf{m}}
\newcommand{\nmsf}{\mathsf{n}}
\newcommand{\pmsf}{\mathsf{p}}
\newcommand{\xmsf}{\mathsf{x}}
\newcommand{\qmbf}{\mathbf{q}}

\newcommand{\Amc}{\mathcal{A}}

\newcommand{\Cmc}{\mathcal{C}}
\newcommand{\Dmc}{\mathcal{D}}

\newcommand{\Imc}{\mathcal{I}}
\newcommand{\Jmc}{\mathcal{J}}
\newcommand{\Kmc}{\mathcal{K}}
\newcommand{\Lmc}{{\mathcal{L}}}

\newcommand{\Qmc}{{\mathcal{Q}}}
\newcommand{\Rmc}{{\mathcal{R}}}

\newcommand{\Tmc}{{\mathcal{T}}}

\newcommand{\Xmbf}{{\mathbf{X}}}

\newcommand{\Amsf}{{\mathsf{A}}}
\newcommand{\Bmsf}{{\mathsf{B}}}
\newcommand{\Cmsf}{{\mathsf{C}}}
\newcommand{\Dmsf}{{\mathsf{D}}}

\newcommand{\Rmsf}{{\mathsf{R}}}
\newcommand{\Smsf}{{\mathsf{S}}}

\newcommand{\Umsf}{{\mathsf{U}}}

\newcommand{\Xmsf}{{\mathsf{X}}}
\newcommand{\Ymsf}{{\mathsf{Y}}}

\newcommand{\aczero}{\mathbf{AC}^0}
\newcommand{\ptime}{\mathbf{P}}
\newcommand{\np}{\mathbf{NP}}
\newcommand{\conp}{\mathbf{coNP}}
\newcommand{\exptime}{\mathbf{ExpTime}}
\newcommand{\twoexptime}{\mathbf{2ExpTime}}

\newcommand{\bhtwo}{\mathbf{BH}_2}

\newtheorem{theorem}{Theorem}
\newtheorem{proposition}{Proposition}
\newtheorem{lemma}{Lemma}
\newtheorem{example}{Example}
\newtheorem{definition}{Definition}
\theoremstyle{remark}

\newtheorem{remark}{Remark}

\title{Queries With Exact Truth Values in Paraconsistent Description Logics}
\author{%
Meghyn Bienvenu$^{1,2}$\and
Camille Bourgaux$^3$\and
Daniil Kozhemiachenko$^1$ \\
\affiliations
$^1$Universit\'{e} de Bordeaux, CNRS, Bordeaux INP, LaBRI, UMR 5800\\
$^2$Japanese-French Laboratory for Informatics, CNRS, NII, IRL 2537, Tokyo, Japan\\ 
$^3$DI ENS, ENS, CNRS, PSL University \& Inria, Paris, France\\
\emails
\{meghyn.bienvenu,daniil.kozhemiachenko\}@u-bordeaux.fr,
camille.bourgaux@ens.fr
}
\begin{document}
\allowdisplaybreaks
\maketitle
\begin{abstract}
We present a novel approach to querying classical inconsistent description logic (DL) knowledge bases by adopting a~paraconsistent semantics with the four ‘Belnapian’ values: \emph{exactly true} ($\true$), \emph{exactly false} ($\false$), \emph{both} ($\both$), and \emph{neither} ($\neither$). In contrast to prior studies on paraconsistent DLs, we allow truth value operators in the query language, which can be used to differentiate between answers having contradictory evidence and those having only positive evidence. We present a reduction to classical DL query answering that allows us to pinpoint the precise combined and data complexity of answering queries with values in paraconsistent $\ALCHI$ and its sublogics. Notably, we show that tractable data complexity is retained for Horn DLs. We present a~comparison with repair-based inconsistency-tolerant semantics, showing that the two approaches are incomparable.
\end{abstract}
\section{Introduction\label{sec:introduction}}
Ontology-mediated query answering (OMQA) has been extensively studied over the past fifteen years 
as a promising paradigm for querying incomplete and 
heterogeneous data \cite{DBLP:journals/jods/PoggiLCGLR08,BienvenuOrtiz2O15,DBLP:conf/ijcai/XiaoCKLPRZ18}.
In a nutshell, OMQA enriches the data with an ontology which provides both a convenient vocabulary for query formulation as well as domain knowledge that is exploited when answering queries. Ontologies are typically formulated in 
logic-based languages (description logics, DLs, being a popular choice) and equipped with a first-order logic semantics, whereby a Boolean (`yes or no') query is deemed to hold whenever it is entailed from the logical theory consisting of the data and 
ontology. An important practical concern with (traditional) OMQA is its lack of robustness 
in the presence of contradictory information, 
as every Boolean 
query is entailed from an inconsistent knowledge base.

A prominent approach to tackling data inconsistencies in OMQA 
is to adopt inconsistency-tolerant semantics based upon repairs, 
defined as inclusion-maximal subsets of the data that are consistent with the ontology. 
Arguably the most natural 
repair-based semantics is the $\AR$ semantics \cite{LemboLenzeriniRosatiRuzziSavo2010} that considers those answers that hold w.r.t.\ each repair, 
inspired by analogous semantics for inconsistent databases \cite{ArenasBertossiChomicki1999}.
Other commonly considered repair semantics include the more permissive $\brave$ semantics \cite{BienvenuRosati2013}, which 
only requires an answer to hold in at least one repair, and the more cautious $\IAR$ semantics~\cite{LemboLenzeriniRosatiRuzziSavo2010}, 
which queries 
 the intersection of all repairs. 
Several other repair-based semantics, 
 incorporating closure operations or various kinds of preferences, 
have been explored, see 
\cite{BienvenuBourgaux2016,Bienvenu2020} for an overview of repair-based semantics for DL knowledge bases. 

Paraconsistent logics represent another natural approach to obtaining 
meaningful answers from contradictory knowledge bases. Whereas 
repair-based semantics are defined in terms of the 
consistent subsets of the inconsistent theory, 
paraconsistent logic semantics, often based upon extended sets of truth values, 
makes it possible for classically inconsistent theories to possess models. 
A common approach is to augment the classical set of truth values $\{\true,\false\}$ with two additional elements --- $\both$ (both true and false) and $\neither$ (neither true nor false).\footnote{Some work considers only $\{\true,\both,\false\}$~\cite{ZhangLinWang2010} or adds other truth values~\cite{KaminskiKnorrLeite2015}.} The four values can be interpreted as four kinds of information one can have on a~given assertion $\Amsf(a)$: only be told that $\Amsf(a)$ is true, only be told that $\Amsf(a)$ is false, be told that $\Amsf(a)$ is both true and false, and be told nothing about $\Amsf(a)$.\footnote{The interpretation is due to Dunn and Belnap~\cite{Dunn1976,Belnap1977computer,Belnap1977fourvalued}, whence the values $\true$, $\both$, $\neither$, and $\false$ are sometimes called ‘Belnapian’.} 
The truth and falsity conditions of Boolean connectives $\neg$, $\sqcap$, and $\sqcup$ are then defined as follows: 
\begin{itemize}[noitemsep,topsep=1pt]
\item $\neg\Amsf(a)$ is \emph{true} if $\Amsf(a)$ is false and vice versa;
\item $[\Amsf\sqcap\Bmsf](a)$ is \emph{true} if $\Amsf(a)$ and $\Bmsf(a)$ are true, and \emph{false} if $\Amsf(a)$ or $\Bmsf(a)$ is false;
\item $[\Amsf\sqcup\Bmsf](a)$ is \emph{true} if $\Amsf(a)$ or $\Bmsf(a)$ is true, and \emph{false} if $\Amsf(a)$ and $\Bmsf(a)$ are false.
\end{itemize}

Paraconsistent DLs were first introduced by~\citeauthor{OdintsovWansing2003} \shortcite{OdintsovWansing2003} and have since then been extensively studied. In particular, four-valued counterparts of expressive description logics such as $\mathcal{SHOIN}(\Dmc)$ and $\mathcal{SROIQ}$ were considered~\cite{MaLinLin2006,MaHitzler2009,Maier2010,MaierMaHitzler2013}. 
Moreover, DLs with non-standard propositional connectives (i.e., whose semantics differ from~\cite{Dunn1976,Belnap1977computer,Belnap1977fourvalued}) were studied by~\citeauthor{ZhangXiaoLinvandenBussche2014} \shortcite{ZhangXiaoLinvandenBussche2014}. 
Most work on paraconsistent DLs has focused on standard reasoning tasks, namely, axiom entailment and consistency checking. Paraconsistent OMQA 
has 
received comparatively less attention and to the best of our knowledge
has only been considered 
by~\citeauthor{NguyenSzalas2012} \shortcite{NguyenSzalas2012} and~\citeauthor{ZhouHuangQiMaHuangQu2012} \shortcite{ZhouHuangQiMaHuangQu2012}. 
Moreover, the query language presented in~\cite{NguyenSzalas2012,ZhouHuangQiMaHuangQu2012} has an unfortunate drawback: given a~knowledge base $\Kmc$ and a~concept $\Amsf$, it is impossible to write a~query $\qmbf$ whose set of answers only contains individuals $a$ for which $\Amsf(a)$ is \emph{exactly true} (i.e., has value $\true$). 
Indeed, we observe (Proposition \ref{prop:hornsameasdropneg}) that for Horn DLs, 
existing approaches to paraconsistent query answering 
correspond to simply ignoring negative axioms, 
and thus fail to benefit from the four-valued semantics. 

Our first main contribution is thus to introduce a new query language for paraconsistent DLs that extends the query language of~\cite{ZhouHuangQiMaHuangQu2012} with value operators, enabling us to differentiate between \emph{at least true} and \emph{exactly true} answers to queries. We explore the computational properties of answering such queries and show, using a translation to classical OMQA, that both the data and combined complexity of paraconsistent query answering in Horn description logic ontologies is the same as that of certain answers under the classical OMQA semantics. For expressive DLs, paraconsistent query answering has the same combined complexity as classical OMQA but in some cases has a slightly higher data complexity. Overall our results show that our paraconsistent query language is more computationally well-behaved than repair-based semantics. 

This brings us to our second contribution: a comparison of paraconsistent and repair-based OMQA semantics. Indeed, 
while the two approaches share similar motivations, to the best of our knowledge, the relationship between them has not been explored. 
We present results showing 
that the two approaches are incomparable. 
More precisely, we show that if we consider queries with the $\true$ (exactly true) operator (which being more restrictive are better suited to approximating repair-based semantics), 
then we neither over-approximate 
$\IAR$, nor under-approximate 
$\brave$ and $\CAR$ (a variant of $\AR$ based on closed repairs). 
This incomparability 
result is generally phrased so as to apply to other paraconsistent DL semantics verifying some basic properties. 

Our paper is structured as 
follows. 
In Section~\ref{sec:ALCfourtriangle}, we define the syntax and semantics of a~four-valued version of $\ALCHI$. Sections~\ref{sec:4queries} and~\ref{sec:complexity} are dedicated to syntax and semantics of the queries incorporating Belnapian values and an analysis of their computational properties. In Section~\ref{sec:4valuedvsrepairs}, we formally compare paraconsistent and repair-based semantics and present a general incomparability result. Finally, we conclude in Section~\ref{sec:conclusion} with a short discussion of future work. Omitted proofs are given in the appendix.

\section{Four-Valued $\ALCHI$ and Its Fragments\label{sec:ALCfourtriangle}}
In this section, we provide the syntax and semantics of four-valued DLs, equipped with a new constructor $\triangle$ that was previously added to Belnap-Dunn propositional logic and its first-order expansion by~\citeauthor{SanoOmori2014} \shortcite{SanoOmori2014}, and which can be intuitively interpreted as follows: $\triangle\Amsf(a)$ means that $\Amsf(a)$ is true and $\neg\triangle\Amsf(a)$ that $\Amsf(a)$ is \emph{not true} (as opposed to $\neg\Amsf(a)$ which means that $\Amsf(a)$ is \emph{false}).
\subsubsection*{Syntax} Let $\CN$, $\RN$ and $\IN$ be three mutually disjoint countable sets of concept, role, and individual names, respectively, and let $\RNinv=\RN\cup\{\Rmsf^-\mid\Rmsf\in\RN\}$ be the set of roles and inverse roles. Given a DL language $\Lmc$, an \emph{$\Lmc^\four_\triangle$ knowledge base} (KB) $\Kmc=\langle\Tmc,\Amc\rangle$ consists of a finite set $\Amc$ of concept and role \emph{assertions} of the form $\Amsf(a)$ and $\Rmsf(a,b)$ respectively, with $a,b\in\IN$, $\Amsf\in\CN$, and $\Rmsf\in\RN$, called the \emph{ABox}, and a finite set $\Tmc$ of axioms whose form depends on the DL $\Lmc$, called the \emph{TBox}. An $\ALCHItrianglefour$ TBox contains \emph{role inclusions} of the form $\Smsf\sqsubseteq\Smsf'$ where $\Smsf,\Smsf'\in\RNinv$ and \emph{concept inclusions} of the form $\Cmsf\sqsubseteq\Dmsf$ where $\Cmsf$ and $\Dmsf$ are $\ALCHItrianglefour$ concepts built using the following grammar: 
$$\Cmsf\coloneqq\top\mid\bot\mid\Amsf\mid\neg\Cmsf\mid\triangle\Cmsf\mid\Cmsf\sqcap\Cmsf\mid\Cmsf\sqcup\Cmsf\mid\exists\Smsf.\Cmsf\mid\forall\Smsf.\Cmsf$$
with $\Amsf\in\CN$ and $\Smsf\in\RNinv$. We also write $\Cmsf\equiv\Dmsf$ as a~shorthand for $\{\Cmsf\sqsubseteq\Dmsf,\Dmsf\sqsubseteq\Cmsf\}$. We sometimes use $\bullet$ and $\circ$ to denote binary connectives from $\{\sqcap,\sqcup\}$ and $\Qmc$ and $\overline{\Qmc}$ for quantifiers from $\{\exists,\forall\}$, assuming that $\bullet\neq\circ$ and $\Qmc\neq\overline{\Qmc}$.

An $\Lmc$ KB is defined as an $\Lmc^\four_\triangle$ KB except that it cannot contain $\triangle$. Besides $\ALCHI$, we will consider the following DL languages which are sub-languages of $\ALCHI$: $\mathcal{ALCH}$ has no inverse roles, $\mathcal{ALCI}$ has no role inclusions, $\mathcal{ALC}$ has neither, $\mathcal{ELHI}_\bot$ does not allow $\sqcup$, $\forall$ and $\neg$, and $\mathcal{ELH}_\bot$, $\mathcal{ELI}_\bot$ and $\mathcal{EL}_\bot$ are obtained from $\mathcal{ELHI}_\bot$ by disallowing inverse roles, role inclusions, and both respectively. Finally DL-Lite$_\Rmc$ TBoxes contain role inclusions of the form $\Smsf\sqsubseteq \Smsf'$ or $\Smsf\sqsubseteq\neg\Smsf'$ and concept inclusions of the form $\Bmsf_1\sqsubseteq \Bmsf_2$ or $\Bmsf_1\sqsubseteq \neg \Bmsf_2$ with $\Bmsf\coloneqq \Amsf\mid\exists \Smsf.\top$. $\mathcal{ELHI}_\bot$ and its sub-logics are \emph{Horn DLs} and we call \emph{propositional TBoxes} the TBoxes that do not use the $\exists$ and $\forall$ constructors.
\subsubsection*{Semantics} The semantics of $\Lmc^\four_\triangle$ is defined through interpretations, which differ from classical DL interpretations in that they define both \emph{positive} and \emph{negative extensions} of concepts. A~\emph{$\four$-interpretation} is a~tuple $\Imc=\langle\Delta^\Imc,\cdot^{\Imc_\pmsf},\cdot^{\Imc_\nmsf}\rangle$ with a domain $\Delta^\Imc\neq\emptyset$, and two interpretation functions $\cdot^{\Imc_\pmsf}$ and $\cdot^{\Imc_\nmsf}$ that map each concept name $\Amsf\in\CN$ to $\Amsf^{\Imc_\pmsf}\subseteq\Delta^\Imc$ and $\Amsf^{\Imc_\nmsf}\subseteq\Delta^\Imc$ respectively, each role name $\Rmsf\in\RN$ to $\Rmsf^{\Imc_\pmsf}=\Rmsf^{\Imc_\nmsf}\subseteq \Delta^\Imc\times\Delta^\Imc$ and each individual name $a\in\IN$ to $a^{\Imc_\pmsf}=a^{\Imc_\nmsf}\in\Delta^\Imc$. 
 For role and individual names interpretations, we can thus omit $\pmsf$ and $\nmsf$ and simply write $\Rmsf^\Imc$ and $a^\Imc$.
 The interpretation functions $\cdot^{\Imc_\pmsf}$ and $\cdot^{\Imc_\nmsf}$ are extended to complex $\ALCHItrianglefour$ concepts and roles as follows.
 \begin{align}\label{equ:ALCHI4propositional}
(\Rmsf^-)^{\Imc}&=\{(y,x)\!\mid\!(x,y)\!\in\!\Rmsf^{\Imc}\}\\
{\top^{\Imc_\pmsf}}&=\Delta^\Imc&\top^{\Imc_\nmsf}&=\emptyset\nonumber\\
(\neg\Cmsf)^{\Imc_\pmsf}&=\Cmsf^{\Imc_\nmsf}&(\neg\Cmsf)^{\Imc_\nmsf}&=\Cmsf^{\Imc_\pmsf}\nonumber\\
(\triangle\Cmsf)^{\Imc_\pmsf}&=\Cmsf^{\Imc_\pmsf}&(\triangle\Cmsf)^{\Imc_\nmsf}&=\Delta^{\Imc}\!\setminus\!\Cmsf^{\Imc_\pmsf}\nonumber\\
(\Cmsf\sqcap\Dmsf)^{\Imc_\pmsf}&=\Cmsf^{\Imc_\pmsf}\!\cap\!\Dmsf^{\Imc_\pmsf}&(\Cmsf\sqcap\Dmsf)^{\Imc_\nmsf}&=\Cmsf^{\Imc_\nmsf}\!\cup\!\Dmsf^{\Imc_\nmsf}\nonumber
\end{align}
\begin{align*}
(\forall\Smsf.\Cmsf)^{\Imc_\pmsf}&=\{x\mid\forall y:(x,y)\in\Smsf^\mathcal{I}\Rightarrow y\in\Cmsf^{\Imc_\pmsf}\}\\
(\forall\Smsf.\Cmsf)^{\Imc_\nmsf}&=\{x\mid\exists y:(x,y)\in\Smsf^\mathcal{I}~\&~y\in\Cmsf^{\Imc_\nmsf}\}
\end{align*}
The semantics of the remaining connectives is given by:
\begin{align*}
\Cmsf\sqcup\Dmsf&\coloneqq\neg(\neg\Cmsf\sqcap\neg\Dmsf)&\exists\Smsf.\Cmsf&\coloneqq\neg\forall\Smsf.\neg\Cmsf
&\bot&\coloneqq\neg\top.
\end{align*}
Given a~$\four$-interpretation $\Imc=\langle\Delta^\Imc,
\cdot^{\Imc_\pmsf},\cdot^{\Imc_\nmsf}\rangle$, $a\in\IN$ and a~concept $\Cmsf$, we will say that
\begin{itemize}[noitemsep,topsep=1pt]
\item $\Cmsf(a)$ is \emph{exactly true in $\Imc$} if $a^\Imc\in\Cmsf^{\Imc_\pmsf}\setminus\Cmsf^{\Imc_\nmsf}$;
\item $\Cmsf(a)$ is \emph{both true and false in $\Imc$} if $a^\Imc\in\Cmsf^{\Imc_\pmsf}\cap\Cmsf^{\Imc_\nmsf}$;
\item $\Cmsf(a)$ is \emph{neither true nor false in $\Imc$} if $a^\Imc\notin\Cmsf^{\Imc_\pmsf}\cup\Cmsf^{\Imc_\nmsf}$;
\item $\Cmsf(a)$ is \emph{exactly false in $\Imc$} if $a^\Imc\in\Cmsf^{\Imc_\nmsf}\setminus\Cmsf^{\Imc_\pmsf}$.
\end{itemize}

A $\four$-interpretation $\Imc=\langle\Delta^\Imc,
\cdot^{\Imc_\pmsf},\cdot^{\Imc_\nmsf}\rangle$ \emph{satisfies} an assertion $\Amsf(a)$ (resp.\ $\Rmsf(a,b)$), 
if $a^\Imc\in \Amsf^{\Imc_\pmsf}$ (resp.\ $(a^\Imc, b^\Imc)\in \Rmsf^\Imc$). 
It satisfies a role inclusion $\Smsf\sqsubseteq\Smsf'$ if $\Smsf^\Imc\subseteq{\Smsf'}^\Imc$, 
and it satisfies a concept inclusion $\Cmsf\sqsubseteq\Dmsf$ if $\Cmsf^{\Imc_\pmsf}\subseteq\Dmsf^{\Imc_\pmsf}$. 
We write $\Imc\models_\four\phi$ when $\Imc$ satisfies an assertion or an axiom $\phi$. 
$\Imc$ is a \emph{$\four$-model} of a KB $\Kmc=\langle\Tmc,\Amc\rangle$, denoted $\Imc\models_\four\Kmc$, if $\Imc\models_\four\phi$ for every $\phi\in\Tmc\cup\Amc$. Finally, $\Kmc$ $\four$-entails an assertion or inclusion $\phi$, denoted $\Kmc\models_\four\phi$, if $\Imc\models_\four\phi$ for every $\four$-model $\Imc$ of $\Kmc$.

The semantics of the classical DL $\ALCHI$ is defined using interpretations with a single interpretation function $\Imc=\langle\Delta^\Imc,\cdot^\Imc\rangle$ where $\cdot^\Imc$ behaves as the positive interpretation function $\cdot^{\Imc_\pmsf}$ except that $(\neg\Cmsf)^\Imc=\Delta^\Imc\setminus \Cmsf^\Imc$ (i.e., the negation is defined classically instead of being paraconsistent). 
We use $\Imc\models\Kmc$ to denote that $\Imc$ is a (classical) model of $\Kmc$, and $\Kmc\models\phi$ to denote that $\Kmc$ (classically) entails $\phi$.

Note that four-valued paraconsistent DLs are sometimes defined 
with \emph{four-valued roles}, i.e., possibly $\Rmsf^{\Imc_\pmsf}\neq\Rmsf^{\Imc_\nmsf}$~\cite{MaierMaHitzler2013,ZhangXiaoLinvandenBussche2014}. We do not use four-valued roles in our presentation for two reasons. First, there are several ways to define them (cf.~\cite[\S5]{Drobyshevich2020}), and it would be cumbersome to consider multiple definitions throughout the paper. 
Second, if we were to adopt the approach in~\cite{MaierMaHitzler2013,ZhangXiaoLinvandenBussche2014}, $\Rmsf^{\Imc_\pmsf}$ is used to define both $(\forall\Rmsf.\Cmsf)^{\Imc_\pmsf}$ and $(\forall\Rmsf.\Cmsf)^{\Imc_\nmsf}$, 
making $\Rmsf^{\Imc_\nmsf}$ redundant in $\ALCHItrianglefour$.

\begin{example}\label{example:teachers1}
Assume that a~university created the following knowledge base $\Kmc_\Umsf=\langle\Tmc_\Umsf,\Amsf_\Umsf\rangle$.
\begin{align*}
\Tmc_\Umsf&=\left\{\begin{matrix}
\mathsf{Prf}\equiv\mathsf{Full}\sqcup\mathsf{Asc},&\exists\mathsf{headof}.\mathsf{Chair}\!\sqsubseteq\!\mathsf{Full},\\
\mathsf{Prf}\sqsubseteq\neg\mathsf{Course},&\mathsf{Full}\sqsubseteq\neg\mathsf{Asc}
\end{matrix}\right\}\nonumber\\
\Amc_\Umsf&=\{\mathsf{headof}(\mathbf{ann},\mathbf{AI}),\mathsf{Chair}(\mathbf{AI}),\mathsf{Asc}(\mathbf{ann})\}
\end{align*}
The TBox expresses that there are two kinds of professors ($\mathsf{Prf}$), full and associate professors ($\mathsf{Full}$, $\mathsf{Asc}$), that heads of chairs are full professors, that professors are not courses and full professors not associate professors. The ABox states that Ann is an associate professor and head of the AI chair. 

If $\Kmc_\Umsf$ is interpreted as a~\emph{classical} ($\ALCHI$) KB, $\Kmc_\Umsf$ is \emph{inconsistent}: there is no classical model of $\Kmc_\Umsf$ since Ann cannot be a full professor and an associate professor at the same time. Hence, everything is entailed from $\Kmc_\Umsf$, for example $\Kmc_\Umsf\models \mathsf{Course}(\mathbf{ann})$. 
If $\Kmc_\Umsf$ is interpreted as a~\emph{paraconsistent} ($\ALCHItrianglefour$) KB, however, there are $\four$-models of $\Kmc_\Umsf$ since $\mathbf{ann}^\Imc$ can belong to $\mathsf{Asc}^{\Imc_\pmsf}$ and $\mathsf{Asc}^{\Imc_\nmsf}$. Actually, this is the case in every $\four$-model, i.e., $\Kmc_\Umsf\models_\four\mathsf{Asc}(\mathbf{ann})$ and $\Kmc_\Umsf\models_\four\neg\mathsf{Asc}(\mathbf{ann})$. Using $\four$-interpretations allows us to obtain more meaningful answers from a classically inconsistent KB (for example, $\Kmc_\Umsf\not\models_\four \mathsf{Course}(\mathbf{ann})$). 

In the classical setting, $\mathsf{Full}\sqsubseteq\neg\mathsf{Asc}$ and $\mathsf{Full}\sqcap \mathsf{Asc}\sqsubseteq\bot$ are equivalent. However, this is not the case in the paraconsistent setting: if we replace the former by the latter in $\Tmc_\Umsf$, then $\Kmc_\Umsf$ has no $\four$-models. It is thus important to carefully write the TBox axioms to reflect the intended meaning. In particular, we can define axioms of different strengths. For example, it may be reasonable to assume that courses and professors should be truly disjoint while one can permit contradictions in concepts governing different kinds of professors (e.g., in the situation above, Ann has been recently appointed the head of the AI chair but her promotion to full professor has not been finalised, so the fact that she is both an associate and full professor only indicates a~minor anomaly in $\Kmc_\Umsf$). In this case, however, it is reasonable to add \emph{contrapositives} of axioms with negations (i.e., $\mathsf{Asc}\sqsubseteq\neg\mathsf{Full}$). This will exclude $\four$-interpretations in which $\mathsf{Full}(a)$ is \emph{exactly true} and $\mathsf{Asc}(a)$ is \emph{both true and false}.

Besides replacing $\mathsf{Prf}\!\sqsubseteq\!\neg\mathsf{Course}$ by $\mathsf{Course}\sqcap\mathsf{Prf}\sqsubseteq\bot$, one can enforce disjointness between $\mathsf{Prf}$ and $\mathsf{Course}$ in several ways using our new $\triangle$ operator. First, with $\mathsf{Prf}\sqsubseteq\neg\triangle\mathsf{Course}$. Second, one can stipulate that $\mathsf{Prf}$ and $\mathsf{Course}$ \emph{behave classically}. $\triangle$ allows for the following compact representation of this requirement: $\neg\mathsf{Prf}\equiv\neg\triangle\mathsf{Prf}$ (an alternative representation is $\top\sqsubseteq \mathsf{Prf}\sqcup\neg\mathsf{Prf}$ and $\mathsf{Prf}\sqcap\neg\mathsf{Prf}\sqsubseteq\bot$). It is important to note that classicality is \emph{stronger} than disjointness because the latter permits the existence of some $a$ s.t.\ $\mathsf{Prf}(a)$ is both true and false or neither true nor false in a~$\four$-in\-ter\-pre\-ta\-tion while the former does not.
\end{example}
\subsubsection*{Capturing Different Inclusion Semantics With $\triangle$} 
We recall different interpretations of $\sqsubseteq$ from the literature:

\begin{definition}[Alternative inclusions]\label{def:othersubsumptions}
Let $\Imc=\langle\Delta^\Imc,\cdot^{\Imc_\pmsf},\cdot^{\Imc_\nmsf}\rangle$ be a~$\four$-interpretation and $\Cmsf,\Dmsf$ be two concepts. 
\begin{itemize}[noitemsep,topsep=1pt]
\item $\Cmsf$ is \emph{internally included} in $\Dmsf$ ($\Imc\!\models_\four\!\Cmsf\!\sqsubseteq\!\Dmsf$) iff $\Cmsf^{\Imc_\pmsf}\subseteq\Dmsf^{\Imc_\pmsf}$.
\item $\Cmsf$ is \emph{materially included} in $\Dmsf$ ($\Imc\!\models_\four\!\Cmsf\!\sqsubseteq^\mmsf\!\Dmsf$) iff \mbox{$\Delta^\Imc\setminus\Cmsf^{\Imc_\nmsf}\subseteq\Dmsf^{\Imc_\pmsf}$}.
\item $\Cmsf$ is \emph{strongly included} in $\Dmsf$ ($\Imc\!\models_\four\!\Cmsf\!\sqsubseteq^\pm\!\Dmsf$) iff $\Cmsf^{\Imc_\pmsf}\subseteq\Dmsf^{\Imc_\pmsf}$ and $\Dmsf^{\Imc_\nmsf}\subseteq\Cmsf^{\Imc_\nmsf}$.
\item $\Cmsf$ is \emph{quasi-classically included} in $\Dmsf$ ($\Imc\!\models_\four\!\Cmsf\!\sqsubseteq^\mathsf{qc}\!\Dmsf$) iff $\Cmsf$ is internally, strongly, and materially included in $\Dmsf$.
\end{itemize}
\end{definition}
Internal, material, and strong inclusions were presented in~\cite{MaHitzlerLin2007} and correspond to three eponymous four-valued implications by~\citeauthor{ArieliAvron1996}~\shortcite{ArieliAvron1996,ArieliAvron1998}. The quasi-classical inclusion proposed by~\citeauthor{ZhangXiaoLinvandenBussche2014} \shortcite{ZhangXiaoLinvandenBussche2014} combines all three notions. 
We have chosen to work with internal inclusion but will show how $\triangle$ allows us to reduce the other interpretations of $\sqsubseteq$ to this one. 

\begin{proposition}\label{prop:subsumptionreduction}
For every pair of $\ALCHItrianglefour$ concepts $(\Cmsf, \Dmsf)$ and $\xmsf\in\{\mmsf,\pm,\mathsf{qc}\}$, there is an $\ALCHItrianglefour$ concept inclusion $\phi_\xmsf$ such that for every $\four$-interpretation $\Imc$, it holds that
\begin{align*}
\Imc\models_\four\Cmsf\sqsubseteq^\xmsf\Dmsf&\text{ iff }\Imc\models_\four\phi_\xmsf.
\end{align*}
\end{proposition}
\begin{proof}
For material inclusion, 
$\Imc\models_\four\Cmsf\!\sqsubseteq^\mmsf\!\Dmsf$ iff $\Imc\models_\four\top\sqsubseteq\neg\Cmsf\!\sqcup\!\Dmsf$. For strong inclusion, it is immediate that $\Imc\models_\four\Cmsf\sqsubseteq^\pm\Dmsf$ iff $\Imc\models_\four\Cmsf\sqsubseteq\Dmsf$ and $\Imc\models_\four\neg\Dmsf\sqsubseteq\neg\Cmsf$, whence,
\begin{align*}
\Imc\models_\four\Cmsf\!\sqsubseteq^\pm\!\Dmsf&\text{ iff }\Imc\models_\four\top\!\sqsubseteq\!(\neg\triangle\Cmsf\sqcup\Dmsf)\sqcap(\neg\Cmsf\sqcup\neg\triangle\neg\Dmsf).
\end{align*}
Finally, $\Imc\models_\four\!\Cmsf\!\sqsubseteq^\mathsf{qc}\!\Dmsf$ iff $\Imc\models_\four\!\Cmsf\!\sqsubseteq^\mmsf\!\Dmsf$ and $\Imc\models_\four\!\Cmsf\!\sqsubseteq^\pm\!\Dmsf$, so $\Imc\models_\four\!\Cmsf\!\sqsubseteq^\mathsf{qc}\!\Dmsf$ is equivalent to
\begin{align*}
\Imc\models_\four\top\sqsubseteq(\neg\Cmsf\sqcup \Dmsf)\sqcap(\neg\triangle\Cmsf\sqcup\Dmsf)\sqcap(\neg\Cmsf\sqcup\neg\triangle\neg\Dmsf).&\qedhere
\end{align*}
\end{proof}

The preceding proof shows how $\triangle$ allows us to succinctly capture different forms of inclusions \emph{without the need to introduce additional concept names} (which would complicate later technical constructions, hence the interest in employing $\triangle$). Indeed, while $\triangle$ can be simulated, this requires us to add new concept names: $\neg\triangle \Cmsf$ can be expressed with a fresh $\Cmsf'$ such that $\Cmsf\sqcap \Cmsf'\sqsubseteq \bot$ and $\top\sqsubseteq \Cmsf\sqcup \Cmsf'$ and $\neg\triangle\neg \Dmsf$ can be replaced by $\Dmsf''$ such that $\neg\Dmsf\sqcap \Dmsf''\sqsubseteq \bot$ and $\top\sqsubseteq \neg\Dmsf\sqcup \Dmsf''$.
\subsubsection*{Negation Normal Form (NNF)} $\ALCHItrianglefour$ knowledge bases can be put into negation normal form (NNF) in polynomial time. We will focus on KBs in NNF to establish the complexity of reasoning when translating four-valued KBs to classical KBs (note that 
\cite{MaierMaHitzler2013} performs this transformation of the KB into NNF while translating it). The difference between our work and previous one is the use of the $\triangle$ constructor. 
\begin{definition}\label{def:NNF}
We say that an $\ALCHItrianglefour$~concept $\Cmsf$ is in \emph{negation normal form (NNF)} if $\Cmsf$ is built from concepts $\Amsf$, $\neg\Amsf$, $\triangle\Amsf$, $\triangle\neg\Amsf$, $\neg\triangle\Amsf$, and $\neg\triangle\neg\Amsf$ 
($\Amsf\in\CN$) using binary connectives and quantifiers.
\end{definition}

\begin{proposition}\label{prop:ontologynormalform}
Let $\Tmc$ be an $\ALCHItrianglefour$ TBox. There exists a~TBox $\nnf(\Tmc)$ s.t.\ all concepts occurring in it are in NNF and $\Imc\models_\four\Tmc$ iff $\Imc\models_\four\nnf(\Tmc)$ for any $\four$-interpretation~$\Imc$.
\end{proposition}
\begin{proof}
We define $\nnf(\Tmc)$ as follows: all role inclusions remain as in $\Tmc$; for each concept inclusion $\Cmsf\sqsubseteq\Dmsf$, we apply the following transformations to $\Cmsf$ and $\Dmsf$.
\begin{align*}
\neg\top&\rightsquigarrow\bot&\neg\bot&\rightsquigarrow\top\nonumber\\
\neg\neg\Cmsf&\rightsquigarrow\Cmsf&\neg(\Cmsf\!\circ\!\Dmsf)&\rightsquigarrow\neg\Cmsf\!\bullet\!\neg\Dmsf\nonumber\\
\neg\Qmc\Smsf.\Cmsf&\rightsquigarrow\overline{\Qmc}\Smsf.\neg\Cmsf&\triangle\triangle\Cmsf&\rightsquigarrow\triangle\Cmsf\nonumber\\
\triangle(\Cmsf\circ\Dmsf)&\rightsquigarrow\triangle\Cmsf\circ\triangle\Dmsf&\triangle\Qmc\Smsf.\Cmsf&\rightsquigarrow\Qmc\Smsf.\triangle\Cmsf\nonumber\\
\triangle\neg\triangle\Cmsf&\rightsquigarrow\neg\triangle\Cmsf
\end{align*}
It can be verified using~\eqref{equ:ALCHI4propositional} that the transformations preserve the concept interpretations, which yields the result. 
\end{proof}
\subsubsection*{Reductions Between $\ALCHItrianglefour$ and $\ALCHI$} 
We show that $\ALCHItrianglefour$ and $\ALCHI$ are equally expressive. Using Proposition~\ref{prop:ontologynormalform}, we can construct an embedding of $\four$-valued knowledge bases into the classical ones. The embedding follows the idea from~\cite{MaHitzlerLin2007}: we encode positive and negative interpretations separately. The only difference in our case is that we need to account for~$\triangle$.
\begin{definition}[Classical counterparts]\label{def:triangleembedding}
Let $\Kmc=\langle\Tmc,\Amsf\rangle$ be an $\ALCHItrianglefour$ KB with $\Tmc$ in NNF. We define its \emph{classical counterpart} $\Kmc^\cl=\{\phi^\cl\mid\phi\in\Kmc\}$ as follows.
\begin{align*}
(\Cmsf\sqsubseteq\Dmsf)^\cl&=\Cmsf^\cl\sqsubseteq\Dmsf^\cl&(\Smsf\sqsubseteq\Smsf')^\cl&=\Smsf\sqsubseteq\Smsf'\\
(\Amsf(a))^\cl&=\Amsf^\cl(a)&(\Rmsf(a,b))^\cl&=\Rmsf(a,b)
\end{align*}
where $\Cmsf,\Dmsf$ are $\ALCHItrianglefour$ concepts, $\Smsf,\Smsf'\in\RNinv$, $\Amsf\in\CN$, $\Rmsf\in\RN$. 
For $\Cmsf$ in NNF, $\Cmsf^\cl$ is defined as follows.
\begin{align}
\Amsf^\cl&=\Amsf^+&(\neg\Amsf)^\cl&=\Amsf^-\nonumber\\
(\triangle\Amsf)^\cl&=\Amsf^+&(\neg\triangle\Amsf)^\cl&=\neg\Amsf^+\nonumber\\
(\triangle\neg\Amsf)^\cl&=\Amsf^-&(\neg\triangle\neg\Amsf)^\cl&=\neg\Amsf^-\nonumber\\
(\Cmsf\circ\Dmsf)^\cl&=\Cmsf^\cl\circ\Dmsf^\cl&(\Qmc\Smsf.\Cmsf)^\cl&=\Qmc\Smsf.\Cmsf^\cl\label{equ:triangletranslation}
\end{align}

Let $\Imc=\langle\Delta^\Imc,\cdot^{\Imc_\pmsf},\cdot^{\Imc_\nmsf}\rangle$ be a~$\four$-valued interpretation. The \emph{classical counterpart} $\Imc^\cl=\langle\Delta^{\Imc^\cl},\cdot^{\Imc^\cl}\rangle$ of $\Imc$ is as follows.
\begin{align}
\Delta^{\Imc^\cl}&=\Delta^\Imc\label{equ:interpretationcounterpart}\\
a^{\Imc^\cl}&=a^\Imc&\Rmsf^{\Imc^\cl}&=\Rmsf^\Imc\tag{$a\in\IN$, $\Rmsf\in\RN$}\\
(\Amsf^+)^{\Imc^\cl}&=\Amsf^{\Imc_\pmsf}&(\Amsf^-)^{\Imc^\cl}&=\Amsf^{\Imc_\nmsf}\tag{$\Amsf\in\CN$}
\end{align}
\end{definition}
\begin{restatable}{proposition}{proptriangleembeddingfirst}\label{prop:triangleembedding1}
Let $\Kmc$ be an $\ALCHItrianglefour$ knowledge base in NNF, $\Imc$ a~$\four$-valued interpretation, and $\phi$ a~concept inclusion, role inclusion, or assertion. Then $\Kmc\models_\four\phi$ iff $\Kmc^\cl\models\phi^\cl$ and, moreover, $\Imc\models_\four\Kmc$ iff $\Imc^\cl\models\Kmc^\cl$.
\end{restatable}

For the other direction, we shall exploit the essentially two-valued behaviour of $\triangle$. We use $\chi^\triangle$ to denote the result of putting $\triangle$ in front of every concept name occurring in $\chi$. 

\begin{restatable}{lemma}{lemreductionclassicalparaconsistent}\label{lem:reduction-classical-paraconsistent}
Let $\Kmc=\langle\Tmc,\Amc\rangle$ be an $\mathcal{ALCHI}$ knowledge base in NNF. Then it holds that $\{\cdot^\Imc\mid \Imc\models\Kmc, \Imc=\langle \Delta^\Imc, \cdot^\Imc\rangle\}=\{\cdot^{\Imc^\four_\pmsf}\mid \Imc_\four\models_\four\Kmc^\triangle, \Imc^\four=\langle \Delta^{\Imc^\four}, \cdot^{\Imc^\four_\pmsf}, \cdot^{\Imc^\four_\nmsf}\rangle\}$.
\end{restatable}

The following proposition straightforwardly follows.

\begin{proposition}\label{prop:triangleembedding2}
Let $\Kmc$ be an $\mathcal{ALCHI}$ knowledge base in NNF 
and $\phi$ a~concept inclusion, role inclusion, or assertion. Then $\Kmc\models\phi$ iff $\Kmc^\triangle\models_\four\phi^\triangle$.
\end{proposition}

The next statement follows from Propositions~\ref{prop:triangleembedding1} and~\ref{prop:triangleembedding2} and the complexity of $\ALCHI$ \cite{Tobies01}.
\begin{theorem}\label{theorem:ExpTimecompleteness}
Axiom or assertion entailment in $\ALCHItrianglefour$ is $\exptime$-complete.
\end{theorem}

\subsubsection*{Horn DLs} 
As we saw in Example~\ref{example:teachers1}, $\Amsf\sqcap\Bmsf\sqsubseteq \bot$, $\Amsf\sqsubseteq \neg \Bmsf$ and $\Bmsf\sqsubseteq \neg \Amsf$ have different semantics in the four-valued setting. Hence, to be able to define $\Lmc^\four_\triangle$ KBs that are really paraconsistent for DLs $\Lmc$ that normally have $\bot$ but no negation, such as $\mathcal{ELHI}_\bot$ and its sub-logics, we need to use syntactic variants that may also express `weak disjointness' with $\neg$. 
An ${\mathcal{ELHI}_\neg}^\four_\triangle$ TBox contains inclusions of one of the following forms (extending $\mathcal{ELHI}_\bot$ in normal form):
\begin{align*}
\Smsf\sqsubseteq\Smsf'&&
\Amsf\sqsubseteq\exists\Smsf.\Bmsf&&\exists\Smsf.\Amsf\sqsubseteq\Cmsf&&\Amsf\sqcap \Bmsf\sqsubseteq\Cmsf&&
\Amsf\sqsubseteq\neg\Bmsf
\end{align*}
with $\Smsf,\Smsf'\in\RNinv$, $\Amsf,\Bmsf\in\CN\cup\{\top\}$ and $\Cmsf\in\CN\cup\{\top,\bot\}$. 
We do not include the $\triangle$ operator in this syntax because we can equivalently add $\triangle$ anywhere in the above inclusions without changing the inclusion satisfaction condition, except in 
the case of $\Amsf\sqsubseteq\neg\triangle\Bmsf$, but as mentioned in Example~\ref{example:teachers1}, $\Amsf\sqsubseteq\neg\triangle\Bmsf$ is equivalent to $\Amsf\sqcap\Bmsf\sqsubseteq\bot$. 
We keep the language name in the form of $\Lmc^\four_\triangle$ only for homogeneity. 
We denote by ${\mathcal{EL}_\neg}^\four_\triangle$, ${\mathcal{ELI}_\neg}^\four_\triangle$ and ${\mathcal{ELH}_\neg}^\four_\triangle$ the fragments of ${\mathcal{ELHI}_\neg}^\four_\triangle$ that correspond to $\mathcal{EL}_\bot$, $\mathcal{ELI}_\bot$ and $\mathcal{ELH}_\bot$ respectively. 

It is easily checked 
that for every ${\Lmc_\neg}^\four_\triangle$ KB $\Kmc$ with $\Lmc\in\{\mathcal{ELHI},\mathcal{ELI},\mathcal{ELH},\mathcal{EL}\}$, its classical counterpart $\Kmc^\cl$ is an $\Lmc_\bot$ KB. Indeed, from the definition of ${\mathcal{ELHI}_\neg}^\four_\triangle$, $\Kmc$ is in NNF and does not contain $\triangle$, so $\cdot^\cl$ only adds superscript $+$ on all concept names but those that occur under $\neg$ in inclusions of the form $\Amsf\sqsubseteq\neg \Bmsf$, which become $\Amsf^+\sqsubseteq\Bmsf^-$. It follows that ${\Lmc_\neg}^\four_\triangle$ has the same complexity as $\Lmc_\bot$. 

Note however that Proposition~\ref{prop:subsumptionreduction} does not hold for ${\mathcal{ELHI}_\neg}^\four_\triangle$. Indeed, as already noted by \citeauthor{MaierMaHitzler2013} \shortcite{MaierMaHitzler2013},
material and strong inclusion require non-Horn concept inclusions, with negation on the left-hand side. 

\section{Queries With Exact Truth Values\label{sec:4queries}}
Before introducing our novel approach to querying four-valued DL KBs, let us recall the query language and semantics considered by \citeauthor{ZhouHuangQiMaHuangQu2012} \shortcite{ZhouHuangQiMaHuangQu2012}.

\begin{definition}\label{def:CQ}
Let $\Var$ be a set of variables disjoint from $\IN$ and $\Term=\Var\cup\IN$. A~\emph{conjunctive query} (CQ) has the form $\qmbf\coloneqq\exists y_1\ldots y_m : \varphi$ where $y_1,\ldots, y_m\in\Var$ and $\varphi$ is a conjunction of atoms of the form $\mathsf{R}(t,t')$ or $\Amsf(t)$ with $t,t'\in\Term$, $\Rmsf\in\RN$ and $\Amsf\in\CN$. 
A CQ $\qmbf$ is \emph{Boolean} (BCQ) if no variable occurs in it freely. 

A~KB $\Kmc$~\emph{$\four$-entails} a~BCQ~$\qmbf$ ($\Kmc\models_\four\qmbf$) if for every $\four$-model $\Imc=\langle\Delta^\Imc,\cdot^{\Imc_\pmsf},\cdot^{\Imc_\nmsf}\rangle$ of $\Kmc$, there is a~match $\pi:\Term\mapsto\Delta^\Imc$ such that for every $c\in\IN$, $\pi(c)=c^\Imc$, and for every $\Rmsf(t_1,t_2)$ (resp.\ $\Amsf(t)$) that occurs in $\qmbf$, $(\pi(t_1),\pi(t_2))\in\Rmsf^\Imc$ (resp.\ $\pi(t)\in\Amsf^{\Imc_\pmsf}$). 
\end{definition}
We make an important observation (not explicit in \cite{ZhouHuangQiMaHuangQu2012}), namely, that in the case of Horn DLs, answering CQs under paraconsistent semantics amounts to answering them classically over the `positive' part of the KB obtained by dropping the 
weak 
disjointness axioms. 
Recall that a classical, two-valued, KB $\Kmc$ entails a BCQ $\qmbf$, denoted $\Kmc\models\qmbf$, iff there is a match for $\qmbf$ in every model 
of $\Kmc$. 

\begin{proposition}\label{prop:hornsameasdropneg}
If $\Kmc$ is an~$\ELHInegtrianglefour$ KB and $\Kmc^+$ denotes the $\mathcal{ELHI}$ KB obtained from $\Kmc$ by dropping all inclusions of the form $\Amsf\sqsubseteq \neg \Bmsf$, then for every BCQ $\qmbf$, $\Kmc\models_\four\qmbf$ iff $\Kmc^+\models\qmbf$.
\end{proposition}
\begin{proof}
Assume that $\Kmc\models_\four\qmbf$. 
 If $\Kmc^+$ has no classical model, then $\Kmc^+\models\qmbf$. Otherwise, 
 let $\Imc=\langle\Delta^\Imc,\cdot^\Imc\rangle$ be a model of $\Kmc^+$. 
 Define $\Jmc=\langle\Delta^\Imc,\cdot^{\Imc},\cdot^{\Imc_\nmsf}\rangle$ with $\Amsf^{\Imc_\nmsf}=\Delta^\Imc$ for every $\Amsf\in\CN$. 
 Since $\neg$ only occurs in inclusions of the form $\Amsf\sqsubseteq \neg \Bmsf$ in $\Kmc$, it is easy to check that $\Jmc\models_\four\Kmc$. 
 It follows that $\Jmc\models_\four\qmbf$, which implies the existence of a match for $\qmbf$ in $\Imc$ by construction of $\Jmc$. Hence $\Kmc^+\models\qmbf$. 

 In the other direction, if $\Kmc\not\models_\four\qmbf$, there is a $\four$-model $\Jmc=\langle\Delta^\Jmc,\cdot^{\Jmc_\pmsf},\cdot^{\Jmc_\nmsf}\rangle$ of $\Kmc$ such that $\Jmc\not\models_\four\qmbf$. Let $\Imc=\langle\Delta^\Jmc,\cdot^{\Jmc_\pmsf}\rangle$. 
 Again, it is easy to check that $\Imc\models\Kmc^+$. Thus $\Imc$ is a model of $\Kmc^+$ such that there is no match for $\qmbf$ in $\Imc$, so $\Kmc^+\not\models\qmbf$. 
\end{proof}
The fact that paraconsistent query answering in Horn DLs basically amounts to ignoring possible sources of contradiction provides strong 
motivation for exploring a more expressive query language that better exploits the paraconsistent semantics. 
We propose such a language by introducing four \emph{value operators} corresponding to Belnapian values. 
\begin{definition}[Queries with values]\label{def:CQV}
A~\emph{conjunctive query with values} (CQV) is a CQ whose atoms are of the form $\mathsf{R}(t,t')$, $\Amsf(t)$ or $\Xmbf(\Amsf(t))$ with 
$\Xmbf\in\{\true,\both,\neither,\false\}$. 
A~\emph{Bo\-ole\-an} CQV (BCQV) has no free variable. 
\end{definition}
We illustrate next the intuitive use of value operators.
\begin{example}\label{example:teachers2}
Let $\Kmc'_\Umsf=\langle\Tmc_\Umsf\cup\Tmc',\Amc_\Umsf\cup\Amc'\rangle$ extend $\Kmc_\Umsf$ from~Example~\ref{example:teachers1}. The additional TBox axioms state that one should not be a teaching assistant ($\mathsf{TA}$) and a professor ($\mathsf{Prf}$), that a course should not be a graduate course ($\mathsf{Gr}$) and an obligatory course ($\mathsf{Obl}$) and that every professor teaches some graduate course. Additional ABox assertions give information about the courses (formal verification $\mathbf{fv}$, algorithms $\mathbf{alg}$, logic $\mathbf{log}$, and automata theory $\mathbf{at}$) taught by four persons as well as the position they hold.
\begin{align*}
\Tmc'&=\left\{\begin{matrix}\mathsf{TA}\!\sqsubseteq\!\neg\mathsf{Prf},&\mathsf{Prf}\!\sqsubseteq\!\neg\mathsf{TA},&\mathsf{Prf}\!\sqsubseteq\!\exists\mathsf{teaches}.\mathsf{Gr},\\
\mathsf{Gr}\sqsubseteq\neg\mathsf{Obl},&\mathsf{Obl}\sqsubseteq\neg\mathsf{Gr}\end{matrix}\right\}\nonumber\\
\Amc'&=\left\{\begin{matrix}\mathsf{teaches}(\mathbf{ann},\mathbf{fv}),&\mathsf{teaches}(\mathbf{ann},\mathbf{alg}),\\\mathsf{teaches}(\mathbf{ann},\mathbf{log}),&\mathsf{teaches}(\mathbf{bea},\mathbf{log}),\\\mathsf{teaches}(\mathbf{bea},\!\mathbf{alg}),&\mathsf{Obl}(\mathbf{log}),\mathsf{Gr}(\mathbf{log}),\\\mathsf{Obl}(\mathbf{alg}),\mathsf{Gr}(\mathbf{fv}),&\mathsf{teaches}(\mathbf{claire},\mathbf{at})\\\mathsf{TA}(\mathbf{bea}),&\mathsf{TA}(\mathbf{claire}),\\\mathsf{Asc}(\mathbf{diane})\end{matrix}\right\}
\end{align*}
Now, consider the following queries:
\begin{align*}
\qmbf_1&\coloneqq\mathsf{teaches}(x,y)\wedge\true(\mathsf{Gr}(y))\\
\qmbf_2&\coloneqq\mathsf{teaches}(x,y)\wedge\neither(\mathsf{Gr}(y))\wedge\neither(\mathsf{Obl}(y))\\
\qmbf_3&\coloneqq\mathsf{teaches}(x,y)\wedge\true(\mathsf{TA}(x))\wedge\both(\mathsf{Obl}(y))\\
\qmbf_\four&\coloneqq\exists y:\true(\mathsf{Asc}(x))\wedge\true(\mathsf{Gr}(y))\wedge\mathsf{teaches}(x,y)
\end{align*}
Intuitively, $\qmbf_1$, $\qmbf_2$, and $\qmbf_3$ look for pairs of persons and courses they teach such that: the course is a graduate course ($\qmbf_1$), the kind of course is not specified ($\qmbf_2$), or the person is a teaching assistant and there is contradictory information about the course being obligatory. One can imagine using $\qmbf_2$ and $\qmbf_3$ to curate the university course database: $\qmbf_2$ will find courses for which some information is missing and $\qmbf_3$ (or a simpler version $\both(\mathsf{Obl}(y))$) will find courses for which contradictory information is provided. On the other hand, $\qmbf_1$ will provide answers for which the kind of the course is not contradicted, hence that we presumably can trust even from the uncurated database.

We thus expect $(\mathbf{ann},\mathbf{fv})$ to be the unique answer for $\qmbf_1$, since $\mathbf{alg}$ is not said to be a graduate course and $\mathbf{log}$ is also registered as an obligatory course, which contradicts that it is a graduate course. Regarding $\qmbf_2$, we expect the unique answer $(\mathbf{claire},\mathbf{at})$, since automata theory is the only course about which it is not specified whether it is graduate or obligatory. Finally, we expect that $(\mathbf{bea},\mathbf{log})$ is the unique answer for $\qmbf_3$. Indeed, Bea is the only teaching assistant who teaches logic since we have $\Kmc_\Umsf'\models_\four\neg\mathsf{TA}(\mathbf{ann})$ using the assertion from Example~\ref{example:teachers1} that Ann is an associate professor. Regarding $\qmbf_4$, which asks for associate professors that teach some graduate course, we expect that $\mathbf{diane}$ is the only answer. Indeed, Diane is the only one who is undoubtedly an associate professor (recall from Example~\ref{example:teachers1} that Ann is a~head of a~chair which means that she is supposed to be a~full professor even though she is listed as an associate). Moreover, although no course taught by Diane is mentioned in the ABox, we know that associate professors should teach at least one graduate course. As this is not contradicted, $\mathbf{diane}$ should be the only answer to $\qmbf_\four$.
\end{example}
We now give the formal semantics of CQVs.
\begin{definition}[Atom sets]\label{def:queryatoms}
Let $\At(\qmbf)$ be the set of all atoms occurring in $\qmbf$ and define for $\Xmbf,\mathbf{Y}\in\{\true,\both,\neither,\false\}$:
\begin{align*}
\At^\Xmbf(\qmbf)&=\{\Amsf(t)\mid\Xmbf(\Amsf(t))\in\At(\qmbf)\}\\\
\At^{\mathbf{XY}}(\qmbf)&=\At^\Xmbf(\qmbf)\cup\At^\mathbf{Y}(\qmbf)\\
\At^+(\qmbf)&=\{\Amsf(t)\mid\Amsf(t)\in\At(\qmbf)\}\cup\At^{\true\both}(\qmbf).
\end{align*}
\end{definition}
\begin{definition}[Answers]\label{def:querysanswer}
A~KB $\Kmc$~\emph{$\four$-entails} a~BCQV~$\qmbf$ ($\Kmc\models_\four\qmbf$) if the following conditions hold.
\begin{enumerate}[noitemsep,topsep=1pt]
\item\label{item:classicalconditionquery} For every $\four$-model $\Imc=\langle\Delta^\Imc,\cdot^{\Imc_\pmsf},\cdot^{\Imc_\nmsf}\rangle$ of $\Kmc$, there is a~match $\pi:\Term\mapsto\Delta^\Imc$ such that for every $c\in\IN$, $\pi(c)=c^\Imc$, and 
\begin{itemize}[noitemsep,topsep=1pt]
\item $(\pi(t_1),\pi(t_2))\in\Rmsf^\Imc$ for every $\Rmsf(t_1,t_2)\in\At(\qmbf)$;
\item $\pi(t)\in\Amsf^{\Imc_\pmsf}$ for every $\Amsf(t)\in\At^+(\qmbf)$;
\item $\pi(t)\in\Amsf^{\Imc_\nmsf}$ for every $\Amsf(t)\in\At^{\both\false}(\qmbf)$.
\end{itemize} 
\item\label{item:valueconditionquery} There exists a~$\four$-model $\Imc=\langle\Delta^\Imc,\cdot^{\Imc_\pmsf},\cdot^{\Imc_\nmsf}\rangle$ of $\Kmc$ and a~match $\pi$ as required above which is additionally s.t.\ 
\begin{itemize}[noitemsep,topsep=1pt]
\item $\pi(t)\notin\Amsf^{\Imc_\nmsf}$ for every $\Amsf(t)\in\At^{\true\neither}(\qmbf)$; 
\item $\pi(t)\notin \Amsf^{\Imc_\pmsf}$ for every $\Amsf(t)\in\At^{\false\neither}(\qmbf)$.
\end{itemize}
\end{enumerate}
We say that $\vec{a}$ is a~four-valued paraconsistent answer to a~CQV $\qmbf(\vec{x})$ with free variables $\vec{x}$ over $\Kmc$ ($\vec{a}\in\ans_\four(\qmbf(\vec{x}),\Kmc)$) if $\Kmc\models_\four\qmbf(\vec{a}) $ where $\qmbf(\vec{a})$ is the Boolean query obtained by replacing the variables from $\vec{x}$ with the constants from $\vec{a}$.
\end{definition}
When $\qmbf$ is just a CQ, the semantics coincides with the one given in Definition~\ref{def:CQ}. 
Indeed, in this case $\At^{\both\false}(\qmbf)$, $\At^{\true\neither}(\qmbf)$ and $\At^{\false\neither}(\qmbf)$ are empty so the condition reduces to item~\ref{item:classicalconditionquery} restricted to its first two points.

One can interpret value operators as follows: $\Kmc\models_\four\true(\Amsf(a))$ means that there is sufficient evidence to conclude $\Amsf(a)$ from $\Kmc$ and no evidence for $\neg\Amsf(a)$; dually, if $\Kmc\models_\four\false(\Amsf(a))$, then we can conclude $\neg\Amsf(a)$ from $\Kmc$ but cannot derive $\Amsf(a)$; $\Kmc\models_\four\both(\Amsf(a))$ means that the evidence regarding $\Amsf(a)$ is \emph{contradictory}; finally, if $\Kmc\models_\four\neither(\Amsf(a))$, then we do not have sufficient information to conclude that $\Amsf(a)$ is true nor to conclude that it is false. 
Intuitively condition~\ref{item:valueconditionquery} in Definition~\ref{def:querysanswer} considers the “negative support” of the query atoms. This allows for distinction between $\Amsf(a)$ being \emph{exactly true} and \emph{at least true} (i.e., true and \emph{maybe false}) and likewise between \emph{exactly false} and \emph{at least false}. We will see in Example~\ref{example:teachers3} that it is impossible to achieve without value operators.

A~straightforward check of the KB and queries in Example~\ref{example:teachers2} now gives the expected answers:
\begin{align*}
\ans_\four(\qmbf_1(x,y),\Kmc'_\Umsf)&=\{(\mathbf{ann},\mathbf{fv})\}\\
\ans_\four(\qmbf_2(x,y),\Kmc'_\Umsf)&=\{(\mathbf{claire},\mathbf{at})\}\\
\ans_\four(\qmbf_3(x,y),\Kmc'_\Umsf)&=\{(\mathbf{bea},\mathbf{log})\}\\
\ans_\four(\qmbf_4(x),\Kmc'_\Umsf)&=\{(\mathbf{diane})\}
\end{align*}
This example illustrates that value operators allow for a~compact and intuitive representation of queries such as ‘a~person who teaches an unspecified course’, or ‘a~person who teaches a~graduate-level course’ (meaning a~course that is labelled as a~graduate-level one \emph{without contradiction}).
\begin{remark}
When used over existentially quantified variables, the semantics of the value operators remains quite lax. 
 Consider for example $\Tmc=\{\Bmsf\sqsubseteq\neg\Amsf\}$, $\Amc=\{\Rmsf(a,b), \Amsf(b),\Bmsf(b)\}$ and $\qmbf=\exists x : \Rmsf(a,x)\wedge \true(\Amsf(x))$. It holds that $\Kmc\models_\four \qmbf$ because every $\four$-model of $\Kmc$ is such that $(a^\Imc,b^\Imc)\in \Rmsf^\Imc$ and $b^\Imc\in\Amsf^{\Imc_\pmsf}$, satisfying item~\ref{item:classicalconditionquery} of Definition~\ref{def:querysanswer}, and there exists a $\four$-model $\Jmc$ of $\Kmc$ such that $(a^\Imc,x)\in \Rmsf^\Jmc$ and $x\notin\Amsf^{\Jmc_\nmsf}$ for some $x\neq b^\Jmc$, satisfying item~\ref{item:valueconditionquery}. Value operators are thus intended to be used preferentially on answer variables or constants. 
\end{remark}

We conclude by briefly discussing alternative semantics we could consider for CQVs and why they are not satisfactory. 
First, if we drop item \ref{item:valueconditionquery} from Definition~\ref{def:querysanswer}, then the semantics of the value operators is overly permissive. For example, $\true(\Amsf(a))$, $\false(\Amsf(a))$, $\both(\Amsf(a))$ and $\neither(\Amsf(a))$ would all be entailed from $\langle\{\Bmsf\sqsubseteq \neg \Amsf\},\{\Amsf(a), \Bmsf(a)\}\rangle$. 
If instead we adopt a naive “certain answers semantics” by considering that $\true(\Amsf(a))$ (resp.\ $\false(\Amsf(a))$, $\both(\Amsf(a))$, $\neither(\Amsf(a))$) is entailed if every model of the KB is such that $\Amsf(a)$ is exactly true (resp.\ exactly false, both true and false, neither true nor false), then the semantics of the value operators is too strict. For example, $\true(\Amsf(a))$ would then not be entailed by $\{\Amsf(a)\}$ and an empty TBox because there are $\four$-models of this KB such that $a^\Imc$ is both in $\Amsf^{\Imc_\pmsf}$ and $\Amsf^{\Imc_\nmsf}$.

\subsubsection*{Relationship to Classical BCQ Entailment} 
We now briefly show how BCQV entailment from a four-valued KB and classical BCQ entailment can be related. Given an~$\ALCHItrianglefour$ KB $\Kmc$ and a BCQV $\qmbf$ such that the only value operators in $\qmbf$ are $\true$ and $\false$, let $\Kmc^\flat$ and $\qmbf^\flat$ be the results of removing all occurrences of $\triangle$ in $\Kmc$ and replacing every $\true(\Amsf(t))$ and $\false(\Amsf(t))$ in $\qmbf$ by $\Amsf(t)$ and $\neg \Amsf(t)$ respectively. 
The query semantics is sound in the following sense. 
\begin{restatable}{proposition}{propsoundness}\label{prop:soundness}
$\Kmc\models_\four\qmbf$ implies $\Kmc^\flat\models\qmbf^\flat$. 
\end{restatable}
The converse holds in a restricted setting. 
\begin{restatable}{proposition}{propcompleteness}\label{prop:completeness}
If $\Kmc$ is a classically satisfiable~$\ELHInegtrianglefour$ KB and 
$\false$ does not occur in $\qmbf$, 
then $\Kmc^\flat\models\qmbf^\flat$ implies $\Kmc\models_\four\qmbf$. 
\end{restatable}
This ensures that when the KB is classically satisfiable, the paraconsistent answers to $\qmbf$ are the same as the classical answers of $\qmbf^\flat$, which is intuitively a desirable property. It does not hold if $\sqcup$ is present, even for assertion entailment, as shown by the following example.
\begin{example}\label{example:consistent}
Let $\Amc=\{\Cmsf(a)\}$ and $\Tmc=\{\Cmsf\sqsubseteq \neg \Bmsf, \Cmsf\sqsubseteq \Amsf\sqcup \Bmsf\}$. $\Kmc$ is consistent and $\Kmc\models \Amsf(a)$. However, $\Kmc\not\models_\four \Amsf(a)$ because of the following $\four$-model of $\Kmc$: $\Amsf^{\Imc_\pmsf}=\Amsf^{\Imc_\nmsf}=\emptyset$, $\Bmsf^{\Imc_\pmsf}=\Bmsf^{\Imc_\nmsf}=\{a^\Imc\}$, $\Cmsf^{\Imc_\pmsf}=\{a^\Imc\}$ and $\Cmsf^{\Imc_\nmsf}=\emptyset$. 
\end{example}
Alternative paraconsistent logics have been proposed to address this arguably counter-intuitive behaviour. For example, \cite{ZhangXiaoLinvandenBussche2014} 
propose a~\emph{strong interpretation} of disjunction (we denote it $\sqcup^\mathsf{qc}$) which allows for the disjunctive syllogism that fails for $\sqcup$. However, it also behaves in a~non-standard manner as $\Amsf(a)\not\models_\four(\Amsf\sqcup^\mathsf{qc}\Bmsf)(a)$. In general, it is unavoidable that paraconsistent logic shows some unexpected behaviour when compared to classical semantics since its basis is to reject some traditional inference principles in order to be able to cope with contradictory information. 
Regarding the second restriction of Proposition~\ref{prop:completeness}, the following example illustrates the issue with $\false$.
\begin{example}
 Let $\Amc=\{\Amsf(a), \Cmsf(a)\}$ and $\Tmc=\{\Amsf\sqsubseteq \exists \Rmsf.\top, \exists \Rmsf.\Bmsf\sqsubseteq \Bmsf, \Bmsf\sqsubseteq \neg \Cmsf, \Cmsf\sqsubseteq \neg \Bmsf\}$, and assume that $\qmbf=\exists x :\Rmsf(a,x)\wedge \false(\Bmsf(x))$, i.e., $\qmbf^\flat=\exists x :\Rmsf(a,x)\wedge \neg \Bmsf(x)$. 
 $\Kmc=\Kmc^\flat$ is consistent and $\Kmc^\flat\models \qmbf^\flat$ but $\Kmc\not\models_\four \qmbf$. 
 Indeed, the following $\four$-interpretation $\Imc$ is such that $\Imc\models_\four\Kmc$ but there is no match for $\qmbf$ in $\Imc$ as required by item \ref{item:classicalconditionquery} of Definition~\ref{def:querysanswer}: 
 $\Rmsf^{\Imc}=\{(a^\Imc, e)\}$, $\Amsf^{\Imc_\pmsf}=\{a^\Imc\}$, $\Amsf^{\Imc_\nmsf}=\emptyset$, $\Bmsf^{\Imc_\pmsf}=\{a^\Imc, e\}$, $\Bmsf^{\Imc_\nmsf}=\{a^\Imc\}$, $\Cmsf^{\Imc_\pmsf}=\{a^\Imc\}$ and $\Cmsf^{\Imc_\nmsf}=\{a^\Imc,e\}$. 
\end{example}

\subsubsection*{Comparison With Other Query Languages} 
We now compare our query language with those proposed in the literature on paraconsistent DLs. As already mentioned, CQVs extend CQs and their semantics is compatible with the one considered by~\citeauthor{ZhouHuangQiMaHuangQu2012} \shortcite{ZhouHuangQiMaHuangQu2012}. \citeauthor{NguyenSzalas2012} \shortcite{NguyenSzalas2012} consider ground queries defined as conjunction of complex assertions of the form $\Cmsf(a)$ (with $\Cmsf$ a potentially complex concept), $\Rmsf(a,b)$, $\neg\Rmsf(a,b)$ and $a\neq b$, interpreted in the expected manner. In particular, $\neg\Rmsf(a,b)$ is entailed from $\Kmc$ if $(a^\Imc,b^\Imc)\in\Rmsf^{\Imc_\nmsf}$ (\citeauthor{NguyenSzalas2012} use four-valued roles). Even if CQVs do not allow directly for the use of $\ALCHItrianglefour$ complex concepts, it is always possible to introduce a concept name $\Amsf$ and add $\Cmsf\equiv\Amsf$ and $\neg\Cmsf\equiv\neg\Amsf$ to the TBox. This will ensure that $\Amsf^{\Imc_\pmsf}=\Cmsf^{\Imc_\pmsf}$ and $\Amsf^{\Imc_\nmsf}=\Cmsf^{\Imc_\nmsf}$.

One can see that $\qmbf_2$ from Example~\ref{example:teachers2} does not have an analogue in the languages of~\cite{ZhouHuangQiMaHuangQu2012} and~\cite{NguyenSzalas2012} since they cannot express that ‘$\Amsf(a)$ \emph{is not true}’ or ‘$\Amsf(a)$ \emph{is not false}’ which is required for the $\neither$ operator. $\both(\Amsf(a))$, on the other hand, can be expressed as $\Amsf(a)\wedge\neg\Amsf(a)$ in the language of~\citeauthor{NguyenSzalas2012}. Note however that this cannot be expressed with the CQs considered by \citeauthor{ZhouHuangQiMaHuangQu2012}, and since they consider DL-Lite ontologies, they cannot either use $\Amsf(a)\wedge\Amsf'(a)$ and a~definition $\neg\Amsf\equiv\Amsf'$ to capture it. 

The inability to express things such as ‘$\Amsf(a)$ \emph{is not true}’ or ‘$\Amsf(a)$ \emph{is not false}’ prevents these query languages from expressing $\true$ and $\false$. The following example illustrates the impact of omitting $\true$ in queries of Example~\ref{example:teachers2}.
\begin{example}\label{example:teachers3}
Consider the following queries
\begin{align*}
\qmbf^\flat_1&\coloneqq\mathsf{teaches}(x,y)\wedge\mathsf{Gr}(y)\\
\qmbf^\flat_4&\coloneqq\exists y:\mathsf{Asc}(x)\wedge\mathsf{Gr}(y)\wedge\mathsf{teaches}(x,y)
\end{align*}
It is clear that $(\mathbf{bea},\mathbf{log})\in\ans_\four(\qmbf^\flat_1(x,y))$ and $\mathbf{ann}\in\ans_\four(\qmbf^\flat_4(x))$. However, it would be problematic as there is an obvious contradiction considering $\mathbf{log}$, whence one cannot be sure whether logic counts as a~graduate or obligatory course. Thus, it might happen that Bea does not teach any graduate courses. Likewise, $\Kmc'_\Umsf$ contains a~contradiction w.r.t.\ Ann's position, whence, it is unclear whether she is still an associate professor or already a~full professor.
\end{example}

\section{Complexity of Query Answering\label{sec:complexity}}
In this section, we establish the complexity of answering CQVs. We do this by constructing a~reduction of CQV answering to answering union of conjunctive queries (UCQs) over classically interpreted knowledge bases.
\begin{definition}\label{def:querytranslation}
Let $\qmbf=\exists\vec{y}:\varphi$ be a~Boolean CQV and let further $\IN_\qmbf=\{c_x\mid x\in\Var\text{ occurs in }\qmbf\}$. Define
\begin{align*}
c_t&=\begin{cases}
t\text{ if }t\in\IN\\
c_t\text{ if }t\in\Var 
\end{cases}
\end{align*}
Using sets of atoms from Definition~\ref{def:queryatoms}, we set
\begin{align*}
\qmbf^+&\coloneqq\exists\vec{y}:\bigwedge\limits_{\Rmsf(t,t')\in\At(\qmbf)}\!\!\!\!\!\!\Rmsf(t,t')\wedge\!\!\!\!\!\!\bigwedge\limits_{\Amsf(t)\in\At^+(\qmbf)}\!\!\!\!\!\!\!\!\!\Amsf^+(t)\wedge\!\!\!\!\!\!\bigwedge\limits_{\Amsf(t)\in\At^{\both\false}(\qmbf)}\!\!\!\!\!\!\!\!\!\Amsf^-(t)\\
\qmbf^\ctr&\coloneqq\bigvee\limits_{\Amsf(t)\in\At^{\true\neither}(\qmbf)}\Amsf^-(c_t)\vee\bigvee\limits_{\Amsf(t)\in\At^{\false\neither}(\qmbf)}\Amsf^+(c_t)\\
\Amc_\qmbf&\coloneqq\{\Rmsf(c_t,c_{t'})\mid\Rmsf(t,t')\in\At(\qmbf^+)\}\cup\\&\hspace{1.5em}\{\Amsf^+(c_t)\mid\Amsf^+
(t)\in\At(\qmbf^+)\}\cup\\&\hspace{1.5em}\{\Amsf^-(c_t)\mid\Amsf^-(t)\in\At(\qmbf^+)\}
\end{align*}
\end{definition}

We are now ready to state our main result.
\begin{restatable}{theorem}{theoremqueryreduction}\label{theorem:queryreduction}
Let $\Kmc$ be an $\ALCHItrianglefour$ KB and $\qmbf$ be a~BCQV. 
\begin{align*}
\Kmc\models_\four\qmbf&\text{ iff }\Kmc^\cl\models\qmbf^+\text{ and }\Kmc^\cl\cup\Amc_\qmbf\not\models\qmbf_\ctr
\end{align*}
\end{restatable}
Intuitively, $\Kmc^\cl\models\qmbf^+$ ensures that the positive interpretation of every $\four$-model of $\Kmc$ satisfies item~\ref{item:classicalconditionquery} of Definition~\ref{def:querysanswer} and $\Kmc^\cl\cup\Amc_\qmbf\not\models\qmbf_\ctr$ ensures that there exists a $\four$-model of $\Kmc$ as required by item~\ref{item:valueconditionquery}. 
Indeed, $\Amc_\qmbf$ enforces a match for $\qmbf^+$ and $\qmbf_\ctr$ checks whether it implies some contradiction of the conditions given by item~\ref{item:valueconditionquery}. The proof relies on classical counterparts and $\four$-counterparts to go from $\four$-models of $\Kmc$ to classical models of $\Kmc^\cl$ or $\Kmc^\cl\cup\Amc_\qmbf$ and vice-versa.

\begin{table}
\begin{tabular*}{\columnwidth}{l @{\extracolsep{\fill}} c c}
\toprule
 & Combined & Data \\
\midrule
$\mathcal{ALCI}, \mathcal{ALCHI}$ & $\twoexptime$-c. & $\conp$-c.
\\
$\mathcal{ALC}$, $\mathcal{ALCH}$& $\exptime$-c. & $\conp$-c.
\\
$\mathcal{ELI}_\bot, \mathcal{ELHI}_\bot$ &$\exptime$-c. & $\ptime$-c.
\\
$\mathcal{EL}_\bot, \mathcal{ELH}_\bot$&$\np$-c. & $\ptime$-c.
\\
DL-Lite$_\Rmc$&$\np$-c. & $\aczero$
\\
\bottomrule
\end{tabular*}
\caption{Complexity of BUCQ entailment over classical KBs. See surveys~\protect\cite{BienvenuOrtiz2O15} for $\mathcal{ELHI}_\bot$ and its sublogics and \protect\cite{DBLP:conf/rweb/OrtizS12} for $\mathcal{ALC}$ and its extensions. 
}\label{tab:complexity-classical}
\end{table}

\begin{table}
\begin{tabular*}{\columnwidth}{l @{\extracolsep{\fill}} c c}
\toprule
 & Combined & Data \\
\midrule
$\mathcal{ALCI}^\four_\triangle, \mathcal{ALCHI}^\four_\triangle$ & $\twoexptime$-c. & $\bhtwo$-c.
\\
$\mathcal{ALC}^\four_\triangle$, $\mathcal{ALCH}^\four_\triangle$ &$\exptime$-c. & $\bhtwo$-c.
\\
${\mathcal{ELI}_\neg}^\four_\triangle, \ELHInegtrianglefour$&$\exptime$-c. & $\ptime$-c.
\\
${\mathcal{EL}_\neg}^\four_\triangle, {\mathcal{ELH}_\neg}^\four_\triangle$&$\np$-c. & $\ptime$-c.
\\
${\text{DL-Lite}_\Rmc}^\four_\triangle$&$\np$-c. & $\aczero$
\\
\bottomrule
\end{tabular*}
\caption{Complexity of BCQV entailment over four-valued KBs.}\label{tab:complexity-paracons}
\end{table}

Using 
Theorem \ref{theorem:queryreduction} and the complexity results for classical DL KBs recalled in Table~\ref{tab:complexity-classical}, we obtain tight complexity results for BCQV entailment in $\mathcal{ALCHI}^\four_\triangle$ and its sublogics, showing that answering queries with values over paraconsistent KBs is often not harder than standard BCQ answering. The only case where we note a complexity increase is 
the data complexity of $\mathcal{ALC}$ and its extensions. 

\begin{theorem}\label{complexity}
The results stated in Table \ref{tab:complexity-paracons} hold. 
\end{theorem}
\begin{proof}
By Theorem \ref{theorem:queryreduction}, $\Kmc\models_\four\qmbf$ iff $\Kmc^\cl\models\qmbf^+$ and $\Kmc^\cl\cup\Amc_\qmbf\not\models\qmbf_\ctr$ so if BUCQ (Boolean union of conjunctive queries) entailment over classical $\Lmc$ KBs is in a complexity class $\Cmc$, BCQV entailment over $\Lmc^\four_\triangle$ KBs can be decided by a Turing machine with a $\Cmc$-oracle (making one~$\Cmc$-call and one~co-$\Cmc$-call). 
Recall that if $\Kmc$ is an ${\Lmc_\neg}^\four_\triangle$ KB with $\Lmc\in\{\mathcal{ELHI},\mathcal{ELI},\mathcal{ELH},\mathcal{EL}\}$, $\Kmc^\cl$ is an $\Lmc_\bot$ KB. 
Moreover, $\qmbf_\ctr$ is actually a disjunction of at most $2|\qmbf|$ assertions and in the case of Horn DLs, $\Kmc^\cl\cup\Amc_\qmbf\not\models\qmbf_\ctr$ iff $\Kmc^\cl\cup\Amc_\qmbf\not\models \alpha$ for every assertion $\alpha$ that occurs in $\qmbf_\ctr$. Since assertion entailment can be done in polynomial time w.r.t.\ combined complexity for $\mathcal{EL}_\bot$, $\mathcal{ELH}_\bot$ and DL-Lite$_\Rmc$ \cite{BaaderBL05,CalvaneseGLLR07}, the $\np$-call to decide $\Kmc^\cl\models\qmbf^+$ and the $\ptime$-calls to decide $\Kmc^\cl\cup\Amc_\qmbf\not\models \alpha$ for each $\alpha$ can be combined in a~single $\np$-call.

Lower bounds for $\Lmc$ transfer to $\Lmc^\four_\triangle$ by Lemma~\ref{lem:reduction-classical-paraconsistent}: given an $\Lmc$ KB $\Kmc$ and BCQ $\qmbf$, $\Kmc\models \qmbf$ iff $\Kmc^\triangle\models_\four \qmbf$ (since there exists a match for $\qmbf$ in every model $\Imc$ of $\Kmc$ iff there exists a match for $\qmbf$ in the positive extension of every $\four$-model $\Imc_\four$ of $\Kmc^\triangle$). 
We obtain the remaining $\bhtwo$-lower bound via a reduction from the $\bhtwo$-complete problem SAT-UNSAT.
\end{proof}

\section{Comparison With Repair-Based Semantics \label{sec:4valuedvsrepairs}}
In this section, we compare paraconsistent querying semantics with existing repair-based semantics. When dealing with repair-based semantics, we assume a \emph{classically consistent} TBox, i.e., we assume that if a~KB is inconsistent, it is due to errors in the ABox. For our comparison, we will naturally consider the popular $\AR$ semantics, which deems a tuple to be an answer if it holds w.r.t.\ every repair. We shall further consider repair-based semantics that provide minimal under-approximation and maximal over-approximations of $\AR$ \cite{BienvenuBourgaux2016}: $\IAR$, $\brave$ and $\CAR$. The $\IAR$ semantics retains only the “safest” answers that are true in the intersection of the repairs, while the $\brave$ semantics considers all answers that hold in at least one repair. Finally, the $\CAR$ semantics over-approximates the $\AR$ semantics in a way that is incomparable with $\brave$, by incorporating a closure operation on the ABox. The latter semantics may seem closer in spirit to paraconsistent reasoning where the positive extensions retain all consequences of the axioms. 

The formal definitions of repairs and the considered repair-based semantics follow. Recall that ABox $\Amc$ is called \emph{$\Tmc$-consistent} if the KB $\langle\Tmc,\Amc\rangle$ is (classically) consistent. 
\begin{definition}[Repairs]\label{def:repair}
Let $\Kmc=\langle\Tmc,\Amc\rangle$ and define
\begin{align*}
\close^*_\Tmc(\Amc)&=\left\{\phi\left|\begin{matrix}\phi\text{ assertion s.t. }\langle\Tmc,\Amc'\rangle\models\phi\\\text{ for some }\Tmc\text{-consistent }\Amc'\subseteq\Amc\end{matrix}\right.\right\}
\end{align*}
\begin{itemize}[noitemsep,topsep=1pt]
\item A~\emph{repair of $\Kmc$} is a maximal 
$\Tmc$-consistent subset of $\Amc$. 
\item A~\emph{closed repair of $\Kmc$} is a~$\Tmc$-consistent $\Rmc\subseteq\close^*_\Tmc(\Amc)$ for which there is no $\Tmc$-consistent $\Rmc'\subseteq\close^*_\Tmc(\Amc)$ s.t.\ either (1) $\Rmc\cap\Amc\subsetneq\Rmc'\cap\Amc$ or (2) $\Rmc\cap\Amc=\Rmc'\cap\Amc$ and $\Rmc\subsetneq\Rmc'$.
\end{itemize}
We denote the set of all repairs (resp.\ closed repairs) of $\Kmc$ with $\Rep(\Kmc)$ (resp.\ $\closeRep(\Kmc)$).
\end{definition}
\begin{definition}[Repair semantics]\label{def:repairsemantics} Let $\qmbf$ be a~Boolean CQ.
\begin{itemize}[noitemsep,topsep=1pt]
\item $\Kmc\models_\AR\qmbf$ if $\langle\Tmc,\Amc'\rangle\models\qmbf$ for every $\Amc'\in\Rep(\Kmc)$.
\item $\Kmc\!\models_\brave\!\qmbf$ if $\langle\Tmc,\!\Amc'\rangle\!\models\!\qmbf$ for some~$\Amc'\in\Rep(\Kmc)$.
\item $\Kmc\!\models_\IAR\!\qmbf$ if $\left\langle\Tmc,\!\!\!\!\!\!\bigcap\limits_{\Amc'\in\Rep(\Kmc)}\!\!\!\!\!\!\Amc'\right\rangle\!\models\!\qmbf$.
\item $\Kmc\models_\CAR\qmbf$ if $\langle\Tmc,\Rmc\rangle\models\qmbf$ for every $\Rmc\in\closeRep(\Kmc)$.
\end{itemize}
\end{definition}

We recall the relations between these semantics.
\begin{center}
\begin{tikzpicture}
\node at (0,0) {$\Kmc\models_{\IAR} \qmbf$};
\node at (1.25,0) {$\implies$};
\node at (2.5,0) {$\Kmc\models_{\AR} \qmbf$};
\node[rotate=30] at (3.75,0.2) {$\implies$};
\node[rotate=330] at (3.75,-0.2) {$\implies$};
\node at (5,0.35) {$\Kmc\models_{\brave} \qmbf$};
\node at (4.9,-0.35) {$\Kmc\models_{\CAR} \qmbf$};
\end{tikzpicture}
\end{center}

We start by remarking that $\models_\four$ over-approximates $\models_\brave$ in Horn DLs. 
\begin{theorem}\label{theorem:braveoverapproximation}
If $\Kmc$ is an $\mathcal{ELHI}_\neg$ KB and $\mathbf{q}$ is a~BCQ, then $\Kmc\models_\brave\mathbf{q}$ implies $\Kmc\models_\four\mathbf{q}$.
\end{theorem}
\begin{proof}
 Assume that $\Kmc\models_\brave\qmbf$: there is a~classically consistent subset $\Kmc'\subseteq\Kmc$ such that $\Kmc'\models\qmbf$. By Proposition~\ref{prop:completeness}, $\Kmc'\models_\four\qmbf$
because $\qmbf$ does not contain any value operator. It follows that $\Kmc\models_\four\qmbf$. Indeed, every $\four$-model of $\Kmc$ is a~$\four$-model of~$\Kmc'$ and since $\qmbf$ does not contain any value operator, item~\ref{item:valueconditionquery} of Definition~\ref{def:querysanswer} is vacuously true.
\end{proof}
Note that Theorem~\ref{theorem:braveoverapproximation} and Proposition~\ref{prop:hornsameasdropneg} are a way to see that in Horn DLs, dropping the negative inclusions $\Amsf\sqsubseteq\neg\Bmsf$ provides an over-approximation of $\brave$. 
However, we cannot generalise Theorem~\ref{theorem:braveoverapproximation} beyond Horn DLs. Indeed, recall that $\models_\four$ and $\models$ differ on consistent KBs (cf.\ Example \ref{example:consistent}) for languages with $\sqcup$, while all repair-based semantics coincide with $\models$ on consistent KBs. 

Since Theorem~\ref{theorem:braveoverapproximation} indicates that $\models_\four$ with CQs (without values) is more permissive than $\brave$, a natural idea for bringing closer paraconsistent reasoning and repair-based reasoning is to add $\true$ on query atoms to strengthen the requirements on answers. We quickly observe that in this case, $\models_\four$ no longer over-approximates (in contrast with Theorem~\ref{theorem:braveoverapproximation}) even the safest semantics $\IAR$, while it does not under-approximate the loosest semantics $\brave$ and $\CAR$. For example, consider the following knowledge base: $\Kmc_\mathsf{ic}=\langle\{\Cmsf\sqsubseteq\Amsf,\Cmsf\sqsubseteq\neg\Amsf,\Cmsf\sqsubseteq\neg\Bmsf\},\{\Cmsf(a),\Bmsf(a)\}\rangle$. The only (closed) repair of $\Kmc_\mathsf{ic}$ is $\{\Bmsf(a)\}$ so $\Kmc_\mathsf{ic}\models_\IAR\Bmsf(a)$ while $\Kmc_\mathsf{ic}\not\models_\brave\Cmsf(a)$ and $\Kmc_\mathsf{ic}\not\models_\CAR\Cmsf(a)$. On the other hand, $\Kmc_\mathsf{ic}\not\models_\four\true(\Bmsf(a))$ while $\Kmc_\mathsf{ic}\models_\four\true(\Cmsf(a))$. However, this example relies on the use of a concept name unsatisfiable w.r.t.~the TBox, which may be not so common in practice. We thus next investigate the case of \emph{coherent} KBs, i.e., KBs where all concept names are satisfiable w.r.t.~the TBox. 

We show that even for coherent KBs, answering CQs under repair-based semantics and answering CQVs over paraconsistent DL KBs is incomparable. Moreover, we show this not only for the paraconsistent DLs we study in this paper but for a wider class of such logics. The following definition, inspired by~\cite[Chapter~3]{Gottwald2001} and~\cite[\S1.5.2]{Skurt2020}, will allow us to state our incomparability results in a general setting, by abstracting from the way extensions (and especially negative extensions) of complex concepts are defined. 
\begin{definition}\label{def:normality}~
For a concept $\Cmsf$ and a~$\four$-interpretation $\Imc$, let 
$\Cmsf^{\Imc_\true}=\Cmsf^{\Imc_\pmsf}\setminus\Cmsf^{\Imc_\nmsf}$ and $\Cmsf^{\Imc_\false}=\Cmsf^{\Imc_\nmsf}\setminus\Cmsf^{\Imc_\pmsf}$. We say that
\begin{itemize}
\item a~unary connective $-$ is
\begin{itemize}
\item \emph{\textsc{neg}-normal} if $x\!\in\!\Cmsf^{\Imc_\true}$ implies $x\!\in\!({-}\Cmsf)^{\Imc_\false}$ and $x\!\in\!\Cmsf^{\Imc_\false}$ implies $x\!\in\!({-}\Cmsf)^{\Imc_\true}$;
\item \emph{\textsc{neg}-standard} if $({-}\Cmsf)^{\Imc_\pmsf}=\Delta^\Imc\setminus\Cmsf^{\Imc_\pmsf}$;
\item \emph{paraconsistent} if $(-\Cmsf)^{\Imc_\pmsf}=\Cmsf^{\Imc_\nmsf}$ and there is a~$\four$-interpretation $\Imc'$ s.t.\ $\Amsf^{\Imc'_\pmsf}\cap\Amsf^{\Imc'_\nmsf}\neq\emptyset$ for some $\Amsf\in\CN$;
\end{itemize}
\item a~binary connective $\circledast$ is
\begin{itemize}
\item \emph{\textsc{and}-normal} if $x\in\Cmsf^{\Imc_\true}\cap\Dmsf^{\Imc_\true}$ implies $x\in(\Cmsf\circledast\Dmsf)^{\Imc_\true}$ and $x\in\Cmsf^{\Imc_\false}\cup\Dmsf^{\Imc_\false}$ implies $x\in(\Cmsf\circledast\Dmsf)^{\Imc_\false}$;
\item \emph{\textsc{and}-standard} if $(\Cmsf\circledast\Dmsf)^{\Imc_\pmsf}=\Cmsf^{\Imc_\pmsf}\cap\Dmsf^{\Imc_\pmsf}$;
\item \emph{\textsc{or}-normal} if $x\in\Cmsf^{\Imc_\true}\cup\Dmsf^{\Imc_\true}$ implies $x\in(\Cmsf\circledast\Dmsf)^{\Imc_\true}$ and $x\in\Cmsf^{\Imc_\false}\cap\Dmsf^{\Imc_\false}$ implies $x\in(\Cmsf\circledast\Dmsf)^{\Imc_\false}$;
\item \emph{\textsc{or}-standard} if $(\Cmsf\circledast\Dmsf)^{\Imc_\pmsf}=\Cmsf^{\Imc_\pmsf}\cup\Dmsf^{\Imc_\pmsf}$;
\end{itemize}
\item a~quantifier $\heartsuit_\Smsf$ is
\begin{itemize}
\item \emph{\textsc{all}-normal} if $\forall y((x,y)\in\Smsf^\Imc\Rightarrow y\in\Cmsf^{\Imc_\true})$ implies $x\in(\heartsuit_\Smsf\Cmsf)^{\Imc_\true}$ and $\exists y((x,y)\in\Smsf^\Imc~\&~y\in\Cmsf^{\Imc_\false})$ implies $x\in(\heartsuit_\Smsf\Cmsf)^{\Imc_\false}$;
\item \emph{\textsc{all}-standard} if $(\heartsuit_\Smsf\Cmsf)^{\Imc_\pmsf}=\{x\mid\forall y:(x,y)\in\Smsf^{\Imc}\Rightarrow y\in\Cmsf^{\Imc_\pmsf}\}$;
\item \emph{\textsc{ex}-normal} if $\exists y((x,y)\in\Smsf^\Imc~\&~y\in\Cmsf^{\Imc_\true})$ implies $x\in(\heartsuit_\Smsf\Cmsf)^{\Imc_\true}$ and $\forall y((x,y)\in\Smsf^\Imc\Rightarrow y\in\Cmsf^{\Imc_\false})$ implies $x\in(\heartsuit_\Smsf\Cmsf)^{\Imc_\false}$;
\item \emph{\textsc{ex}-standard} if $(\heartsuit_\Smsf\Cmsf)^{\Imc_\pmsf}=
\{x\mid\exists y:(x,y)\in\Smsf^{\Imc} \& y\in\Cmsf^{\Imc_\pmsf}\}$.
\end{itemize}
\end{itemize}
\end{definition}

Considering $\ALCHItrianglefour$ connectives, $\neg,\sqcap,\sqcup,\exists\Smsf,\forall\Smsf$ are \textsc{neg}-, \textsc{and}-, \textsc{or}-, \textsc{ex}- and \textsc{all}-normal respectively, while $\sqcap,\sqcup,\exists\Smsf,\forall\Smsf$ are \textsc{and}-, \textsc{or}-, \textsc{ex}- and \textsc{all}-standard respectively, and $\neg$ is not \textsc{neg}-standard but paraconsistent. Note that normality and standardness do not imply one another: the strong interpretation of disjunction in~\cite{ZhangXiaoLinvandenBussche2014} is \textsc{or}-normal but not \textsc{or}-standard, while 
$\otimes$ from~\cite{OmoriSano2015} is \textsc{and}-standard but not \textsc{and}-normal.

The next theorem states our incomparability result when value operators are allowed in queries: even for atomic concept queries over coherent Horn DL KBs, 
when we put $\true$ on top of the query atom, $\models_\four$ does not over-approximate 
$\IAR$, while it does not under-approximate 
$\brave$ and $\CAR$. 
\begin{definition}
Given a~set of connectives $\mathfrak{C}=\{-,\otimes,\oplus,\blacksquare_\Smsf,\blacklozenge_\Smsf\}$ and an $\ALCHI$ concept $\Cmsf$, we denote by $\Cmsf^\mathfrak{C}$ the concept obtained from $\Cmsf$ by replacing $\neg$ with $-$, $\sqcap$ with $\otimes$, $\sqcup$ with $\oplus$, $\forall\Smsf$ with~$\blacksquare_\Smsf$, and $\exists\Smsf$ with~$\blacklozenge_\Smsf$. We say that a~query entailment relation $\models_\Ymsf$ \emph{$\true$-over-approximates}
(resp.\ \emph{$\true$-under-approximates}) $\models_\Xmsf$ under $\mathfrak{C}$ if $\Kmc\models_\Xmsf\Amsf(a)$ implies $\Kmc^\mathfrak{C}\models_\Ymsf \true(\Amsf(a))$ (resp. $\Kmc^\mathfrak{C}\models_\Ymsf\true(\Amsf(a))$ implies $\Kmc\models_\Xmsf\Amsf(a)$) 
for any KB $\Kmc$ and Boolean atomic concept query $\Amsf(a)$, where $\Kmc^\mathfrak{C}$ is obtained from $\Kmc$ by replacing every concept $\Cmsf$ by $\Cmsf^\mathfrak{C}$. 
\end{definition}

\begin{theorem}\label{theorem:noapproximations}
It holds that:
\begin{itemize}[noitemsep,topsep=1pt]
\item $\models_\four$ does not $\true$-over-approximate $\models_\IAR$, and 
\item $\models_\four$ does not $\true$-under-approximate $\models_\brave$ and $\models_\CAR$
\end{itemize}
under $\mathfrak{C}$ when $\sqsubseteq$ is the internal inclusion and for
\begin{enumerate}[noitemsep,topsep=1pt]
\item coherent DL-Lite ontologies if $-$ is paraconsistent and $\blacklozenge_\Smsf$ is \textsc{ex}-normal or \textsc{ex}-standard,
\item coherent propositional Horn ontologies if $-$ is paraconsistent and $\otimes$ is \textsc{and}-standard.
\end{enumerate}
\end{theorem}
\begin{proof}
For point 1, consider $\Kmc_1=\langle\Tmc_1,\Amc_1\rangle$ and let $\Kmc^\mathfrak{C}_1=\langle\Tmc^\mathfrak{C}_1,\Amc^\mathfrak{C}_1\rangle$ be the result of replacing $\exists$ with $\blacklozenge$ and $\neg$ with $-$.
\begin{align*}
\Tmc_1&=\{\exists\Rmsf.\top\sqsubseteq\Amsf,\exists\Rmsf^-.\top\sqsubseteq\neg\Amsf, \exists\Rmsf.\top\sqsubseteq\Cmsf, \Cmsf\sqsubseteq\neg\Bmsf\}\\
\Amc_1&=\{\Rmsf(a,a), \Bmsf(a)\}
\end{align*}
Observe that the only (closed) repair is $\{\Bmsf(a)\}$, so $\Kmc_1\models_\IAR\Bmsf(a)$, $\Kmc_1\not\models_\brave\Cmsf(a)$ and $\Kmc_1\not\models_\CAR\Cmsf(a)$. However, we can show that $\Kmc^\mathfrak{C}_1\not\models_\four\true(\Bmsf(a))$ while $\Kmc^\mathfrak{C}_1\models_\four\true(\Cmsf(a))$. 
Indeed, for every $\four$-model $\Imc$ of $\Kmc^\mathfrak{C}_1$, since $(a^\Imc,a^\Imc)\in\Rmsf^\Imc$, by \textsc{ex}-standardness or normality of $\blacklozenge_\Rmsf$, 
$a^\Imc\in(\blacklozenge_\Rmsf\top)^{\Imc_\pmsf}$ so $a^\Imc\in\Cmsf^{\Imc_\pmsf}$ and $a^\Imc\in(-\Bmsf)^{\Imc_\pmsf}=\Bmsf^{\Imc_\nmsf}$ since $-$ is paraconsistent. 
Moreover, the following $\four$-model $\Jmc$ of $\Kmc^\mathfrak{C}_1$ is such that $a^\Jmc\notin\Cmsf^{\Jmc_\nmsf}$: $\Rmsf^\Jmc=\{(a^\Jmc,a^\Jmc)\}$, $\Amsf^{\Jmc_\pmsf}=\Amsf^{\Jmc_\nmsf}=\{a^\Jmc\}$, $\Bmsf^{\Jmc_\pmsf}=\Bmsf^{\Jmc_\nmsf}=\{a^\Jmc\}$, $\Cmsf^{\Jmc_\pmsf}=\{a^\Jmc\}$, and $\Cmsf^{\Jmc_\nmsf}=\emptyset$.

For point 2, we consider the following example inspired by~\cite{BienvenuBourgaux2016}. Let $\Kmc_2=\langle\Tmc_2,\Amc_2\rangle$ and $\Kmc^\mathfrak{C}_2$ be the result of replacing $\sqcap$ by $\otimes$ and $\neg$ by $-$.
\begin{align*}
\Tmc_2&=\left\{\begin{matrix}\mathsf{Prf}\sqsubseteq\mathsf{Emp},&\mathsf{UGr}\sqsubseteq\mathsf{Std},\\
\mathsf{Std}\sqcap\mathsf{Emp}\sqsubseteq\mathsf{EmpStd},&\mathsf{EmpStd}\sqsubseteq\neg\mathsf{Tech},\\
\mathsf{Prf}\sqsubseteq\neg\mathsf{Std},&\mathsf{Std}\sqsubseteq\neg\mathsf{Prf},\\
\mathsf{UGr}\sqsubseteq\neg\mathsf{Emp},&
\mathsf{Emp}\sqsubseteq\neg\mathsf{UGr}&
\end{matrix}\right\}\nonumber\\
\Amc_2&=\{\mathsf{Prf}(\mathbf{s}), \mathsf{UGr}(\mathbf{s}),\mathsf{Tech}(\mathbf{s})\}
\end{align*}
The repairs are $\{\mathsf{Prf}(\mathbf{s}),\mathsf{Tech}(\mathbf{s})\}$ and $\{\mathsf{UGr}(\mathbf{s}),\mathsf{Tech}(\mathbf{s})\}$ and the closed repairs are $\{\mathsf{Prf}(\mathbf{s}),\mathsf{Tech}(\mathbf{s}),\mathsf{Emp}(\mathbf{s})\}$ and $\{\mathsf{UGr}(\mathbf{s}),\mathsf{Tech}(\mathbf{s}), \mathsf{Std}(\mathbf{s})\}$ so $\Kmc_2\models_\IAR\mathsf{Tech}(\mathbf{s})$, $\Kmc_2\not\models_\brave\mathsf{EmpStd}(\mathbf{s})$ and $\Kmc_2\not\models_\CAR\mathsf{EmpStd}(\mathbf{s})$. However, we show that $\Kmc^\mathfrak{C}_2\not\models_\four\true(\mathsf{Tech}(\mathbf{s}))$ while $\Kmc^\mathfrak{C}_2\models_\four\true(\mathsf{EmpStd}(\mathbf{s}))$. Indeed, let $\Imc$ be a $\four$-model of $\Kmc^\mathfrak{C}_2$. Since $\mathbf{s}^\Imc\in\mathsf{Prf}^{\Imc_\pmsf}$, $\mathbf{s}^\Imc\in\mathsf{Emp}^{\Imc_\pmsf}$ and similarly, since $\mathbf{s}^\Imc\in\mathsf{UGr}^{\Imc_\pmsf}$, $\mathbf{s}^\Imc\in\mathsf{Std}^{\Imc_\pmsf}$. Since $\mathbf{s}^\Imc\in\mathsf{Std}^{\Imc_\pmsf}\cap\mathsf{Emp}^{\Imc_\pmsf}$, by \textsc{and}-standardness of $\otimes$, $\mathbf{s}^\Imc\in(\mathsf{Std}\otimes\mathsf{Emp})^{\Imc_\pmsf}$. 
Hence, $\mathbf{s}^\Imc\in\mathsf{EmpStd}^{\Imc_\pmsf}$ and $\mathbf{s}^\Imc\in (-\mathsf{Tech})^{\Imc_\pmsf}=\mathsf{Tech}^{\Imc_\nmsf}$. 
Moreover, the following $\four$-model $\Jmc$ of $\Kmc^\mathfrak{C}_2$ is such that $\mathbf{s}^\Jmc\notin\mathsf{EmpStd}^{\Jmc_\nmsf}$: $\Amsf^{\Jmc_\pmsf}=\Amsf^{\Jmc_\nmsf}=\{\mathbf{s}^\Jmc\}$ for $\Amsf\in\{\mathsf{Prf},\mathsf{Emp},\mathsf{UGr},\mathsf{Std},\mathsf{Tech}\}$, $\mathsf{EmpStd}^{\Jmc_\pmsf}=\{\mathbf{s}^\Jmc\}$ and $\mathsf{EmpStd}^{\Jmc_\nmsf}=\emptyset$.
\end{proof}
One could of course wonder what happens if we use value operators other than $\true$ in queries. 
However, note that $\neither(\Amsf(a))$ and $\false(\Amsf(a))$ require that $a^\Imc\notin\Amsf^{\Imc_\pmsf}$ in all $\four$-models, so intuitively that $\Amsf(a)$ cannot be derived, while queries under repair-based semantics only look for answers that can be derived in some way. 
Regarding $\both$, we can see that $\Kmc^\mathfrak{C}_1\models_\four\both(\Amsf(a))$ while $\Kmc_1\not\models_\Xmsf\Amsf(a)$ with $\Xmsf\in\{\brave,\CAR\}$, and that if we let $\Kmc_3=\langle\emptyset,\{\Dmsf(a)\}\rangle$, $\Kmc^\mathfrak{C}_3\models_\IAR \Dmsf(a)$ while $\Kmc^\mathfrak{C}_3\not\models_\four\both(\Dmsf(a))$.

We conclude this section by recalling the computational advantages of paraconsistent reasoning over repair-based semantics: 
BCQ entailment under $\brave$, (resp.~$\AR$ and $\IAR$) is $\Sigma^P_2$-hard (resp.~$\Pi^P_2$-hard) in $\mathcal{ALC}$ and $\np$-hard (resp.~$\conp$-hard) in $\mathcal{EL}_\bot$ w.r.t.~data complexity \cite{BienvenuBourgaux2016}. Moreover, paraconsistent reasoning does not need to assume that the TBox is satisfiable (the $\AR$ semantics has been generalised to repairs that may remove part of the TBox as well \cite{EiterLukasiewiczPredoiu2016} but the complexity of the generalised semantics is at least as high as that of $\AR$).

\section{Conclusion and Discussion}\label{sec:conclusion}
In this paper, we presented a new approach to querying inconsistent DL KBs based upon paraconsistent logic, which we show to be incomparable to repair-based semantics. Differently from existing paraconsistent OMQA approaches, our query language enables us to take full advantage of the four-valued semantics, making it possible to differentiate between \emph{exactly true} and \emph{at least true} instances of a~concept. We proved that our approach is computationally well-behaved (cf.~Table~\ref{tab:complexity-paracons}): in Horn KBs, the combined and data complexity of paraconsistent query answering coincides with that of the classical certain answers semantics; in expressive DLs, data complexity of paraconsistent BCQV entailment remains lower than in repair-based semantics. 
Moreover, our complexity results rely on a simple reduction of CQV answering to OMQA, providing a way to readily implement CQV answering. We also expect that the technique based on translation we provide can be adapted to more expressive DLs (the translation given by \cite{MaierMaHitzler2013} that we adapted was for $\mathcal{SROIQ}$). 

Paraconsistent DLs assign truth values to concept assertions and are in this regard close to \emph{fuzzy DLs}, in which concept memberships are evaluated using \emph{degrees}. In particular, it is natural to wonder whether there is a relationship between our work and \emph{lattice-based} fuzzy DLs that allow for incomparable membership degrees \cite{DBLP:journals/ijar/BorgwardtP14}. If we consider fuzzy DLs based on a lattice formed by Belnapian values (be it the lattice with $\true$ as the supremum and $\false$ as the infimum or the one with $\both$ as the supremum and $\neither$ as the infimum) and queries that allow one to ask that a concept holds to at least some degree, then one can capture the semantics of CQs (without value operators) in paraconsistent DLs (Definition~\ref{def:CQ}). However, it would not be possible to capture CQVs under the semantics we introduce. For example, considering the lattice with $\both$ as supremum, in our semantics $\true(\Amsf(a))$ would mean that $\Amsf(a)$ has degree at least $\true$ in all models and that \emph{there exists a model such that $\Amsf(a)$ has not degree at least $\false$}, which is not directly expressible in fuzzy DLs.

Following this idea of queries requiring that an atom has “degree at least $\Xmbf$”, note that CQV atoms of the form $\Amsf(t)$ can be seen as two-valued atoms “$\true$ or $\both$” (at least positive evidence). We could extend the definition of CQVs to allow for multi-valued atoms and would easily treat the case “$\false$ or $\both$" (at least negative evidence) by extending $\qmbf^+$ (Definition~\ref{def:querytranslation}) with $\Amsf^-(t)$ for such atoms and not taking them into account in $\qmbf^\ctr$. However, allowing multi-valued query atoms in general would affect the results. For example, the cases “$\neither$ or $\false$" (no positive evidence) or “$\neither$ or $\true$" (no negative evidence) would be equivalent to having a (classical) negation in the query (we would need atoms of the form $\neg \Amsf^+(t)$ or $\neg \Amsf^-(t)$ in $\qmbf^+$) so we would need to reduce queries with such atoms to queries with negative atoms in classical DLs, which are known to be much harder to handle and will lead to higher complexity results. 

In future work, we plan to adapt CQVs to paraconsistent DL knowledge bases with four-valued roles by allowing value operators also on role atoms. We expect that our complexity results will continue to hold in the presence of four-valued roles under any of the previously proposed semantics, by adapting the translation-based approach. However, what kind of new inferences we can obtain by adding four-value roles will depend on the adopted semantics for roles and on which other DL constructors are present.
\section*{Acknowledgements}
This work was supported by the ANR AI Chair INTENDED (ANR-19-CHIA-0014).
\bibliographystyle{kr}
\bibliography{kr-sample}

\begin{thebibliography}{}

\bibitem[\protect\citeauthoryear{Arenas, Bertossi, and Chomicki}{1999}]{ArenasBertossiChomicki1999}
Arenas, M.; Bertossi, L.; and Chomicki, J.
\newblock 1999.
\newblock {Consistent Query Answers in Inconsistent Databases}.
\newblock In {\em {Proceedings of the 18th ACM SIGMOD-SIGACT-SIGAI Symposium on Principles of Database Systems (PODS-1999)}}.

\bibitem[\protect\citeauthoryear{Arieli and Avron}{1996}]{ArieliAvron1996}
Arieli, O., and Avron, A.
\newblock 1996.
\newblock {Reasoning With Logical Bilattices}.
\newblock {\em Journal of Logic, Language and Information} 5(1).

\bibitem[\protect\citeauthoryear{Arieli and Avron}{1998}]{ArieliAvron1998}
Arieli, O., and Avron, A.
\newblock 1998.
\newblock {The Value of the Four Values}.
\newblock {\em Artificial Intelligence} 102(1):97--141.

\bibitem[\protect\citeauthoryear{Baader, Brandt, and Lutz}{2005}]{BaaderBL05}
Baader, F.; Brandt, S.; and Lutz, C.
\newblock 2005.
\newblock {Pushing the {$\mathcal{EL}$} Envelope}.
\newblock In {\em Proceedings of the 19th International Joint Conference on Artificial Intelligence (IJCAI-05)}.

\bibitem[\protect\citeauthoryear{Belnap}{1977a}]{Belnap1977fourvalued}
Belnap, N.
\newblock 1977a.
\newblock {A Useful Four-Valued Logic}.
\newblock In {\em Modern Uses of Multiple-Valued Logic},  5--37.
\newblock Dordrecht: Springer Netherlands.

\bibitem[\protect\citeauthoryear{Belnap}{1977b}]{Belnap1977computer}
Belnap, N.
\newblock 1977b.
\newblock {How a Computer Should Think}.
\newblock In {\em Contemporary Aspects of Philosophy}. Oriel Press.

\bibitem[\protect\citeauthoryear{Bienvenu and Bourgaux}{2016}]{BienvenuBourgaux2016}
Bienvenu, M., and Bourgaux, C.
\newblock 2016.
\newblock {Inconsistency-Tolerant Querying of Description Logic Knowledge Bases}.
\newblock In {\em Reasoning Web: Logical Foundation of Knowledge Graph Construction and Query Answering}, volume 9885 of {\em Lecture Notes in Computer Science},  156--202.
\newblock Springer.

\bibitem[\protect\citeauthoryear{Bienvenu and Ortiz}{2015}]{BienvenuOrtiz2O15}
Bienvenu, M., and Ortiz, M.
\newblock 2015.
\newblock {Ontology-Mediated Query Answering With Data-Tractable Description Logics}.
\newblock In {\em Reasoning Web. Web Logic Rules}, volume 9203 of {\em Lecture Notes in Computer Science},  218--307.
\newblock Springer.

\bibitem[\protect\citeauthoryear{Bienvenu and Rosati}{2013}]{BienvenuRosati2013}
Bienvenu, M., and Rosati, R.
\newblock 2013.
\newblock {Tractable Approximations of Consistent Query Answering for Robust Ontology-based Data Access}.
\newblock In {\em Proceedings of the 23rd International Joint Conference on Artificial Intelligence (IJCAI-2013)}.

\bibitem[\protect\citeauthoryear{Bienvenu}{2020}]{Bienvenu2020}
Bienvenu, M.
\newblock 2020.
\newblock {A Short Survey on Inconsistency Handling in Ontology-Mediated Query Answering}.
\newblock {\em K\"{u}nstliche Intelligenz} 34(4):443--451.

\bibitem[\protect\citeauthoryear{Borgwardt and Pe{\~{n}}aloza}{2014}]{DBLP:journals/ijar/BorgwardtP14}
Borgwardt, S., and Pe{\~{n}}aloza, R.
\newblock 2014.
\newblock {Consistency Reasoning in Lattice-Based Fuzzy Description Logics}.
\newblock {\em International Journal of Approximate Reasoning} 55(9):1917--1938.

\bibitem[\protect\citeauthoryear{Calvanese \bgroup et al\mbox.\egroup }{2007}]{CalvaneseGLLR07}
Calvanese, D.; Giacomo, G.~D.; Lembo, D.; Lenzerini, M.; and Rosati, R.
\newblock 2007.
\newblock {Tractable Reasoning and Efficient Query Answering in Description Logics: The DL-Lite Family}.
\newblock {\em Journal of Automated Reasoning} 39(3):385--429.

\bibitem[\protect\citeauthoryear{Drobyshevich}{2020}]{Drobyshevich2020}
Drobyshevich, S.
\newblock 2020.
\newblock {A General Framework for $\mathsf{FDE}$-Based Modal Logics}.
\newblock {\em Studia Logica} 108(6):1281--1306.

\bibitem[\protect\citeauthoryear{Dunn}{1976}]{Dunn1976}
Dunn, J.
\newblock 1976.
\newblock {Intuitive Semantics for First-Degree Entailments and ‘Coupled Trees’}.
\newblock {\em Philosophical Studies} 29(3):149--168.

\bibitem[\protect\citeauthoryear{Eiter, \L{}ukasiewicz, and Predoiu}{2016}]{EiterLukasiewiczPredoiu2016}
Eiter, T.; \L{}ukasiewicz, T.; and Predoiu, L.
\newblock 2016.
\newblock {Generalized Consistent Query Answering under Existential Rules}.
\newblock In {\em Proceedings of Principles of Knowledge Representation and Reasoning (KR-2016)},  359--368.
\newblock {AAAI} Press.

\bibitem[\protect\citeauthoryear{Gottwald}{2001}]{Gottwald2001}
Gottwald, S.
\newblock 2001.
\newblock {\em {A Treatise on Many-Valued Logics}}, volume~3 of {\em Studies in Logic and Computation}.
\newblock Research studies press.

\bibitem[\protect\citeauthoryear{Kaminski, Knorr, and Leite}{2015}]{KaminskiKnorrLeite2015}
Kaminski, T.; Knorr, M.; and Leite, J.
\newblock 2015.
\newblock {Efficient Paraconsistent Reasoning With Ontologies and Rules}.
\newblock In {\em Proceedings of the 24th International Joint Conference on Artificial Intelligence (IJCAI-2015)},  3098--3105.
\newblock AAAI.

\bibitem[\protect\citeauthoryear{Lembo \bgroup et al\mbox.\egroup }{2010}]{LemboLenzeriniRosatiRuzziSavo2010}
Lembo, D.; Lenzerini, M.; Rosati, R.; Ruzzi, M.; and Savo, D.
\newblock 2010.
\newblock {Inconsistency-Tolerant Semantics for Description Logics}.
\newblock In {\em Web Reasoning and Rule Systems (RR-2010)}, volume 6333 of {\em Lecture Notes in Computer Science},  103--117.
\newblock Berlin, Heidelberg: Springer.

\bibitem[\protect\citeauthoryear{Ma and Hitzler}{2009}]{MaHitzler2009}
Ma, Y., and Hitzler, P.
\newblock 2009.
\newblock {Paraconsistent Reasoning for OWL 2}.
\newblock In {\em Web Reasoning and Rule Systems (RR-2009)}, volume 5837 of {\em Lecture Notes in Computer Science}. Berlin, Heidelberg: Springer Berlin Heidelberg.
\newblock  197--211.

\bibitem[\protect\citeauthoryear{Ma, Hitzler, and Lin}{2007}]{MaHitzlerLin2007}
Ma, Y.; Hitzler, P.; and Lin, Z.
\newblock 2007.
\newblock {Algorithms for Paraconsistent Reasoning With OWL}.
\newblock In {\em The Semantic Web: Research and Applications. (ESWC-2007)}, volume 4519 of {\em Lecture Notes in Computer Science}. Berlin, Heidelberg: Springer Berlin Heidelberg.
\newblock  399--413.

\bibitem[\protect\citeauthoryear{Ma, Lin, and Lin}{2006}]{MaLinLin2006}
Ma, Y.; Lin, Z.; and Lin, Z.
\newblock 2006.
\newblock {Inferring With Inconsistent OWL DL Ontology: A Multi-Valued Logic Approach}.
\newblock In {\em Current Trends in Database Technology (EDBT-2006)}, volume 4254 of {\em Lecture notes in computer science}. Berlin, Heidelberg: Springer Berlin Heidelberg.
\newblock  535--553.

\bibitem[\protect\citeauthoryear{Maier, Ma, and Hitzler}{2013}]{MaierMaHitzler2013}
Maier, F.; Ma, Y.; and Hitzler, P.
\newblock 2013.
\newblock {Paraconsistent OWL and Related Logics}.
\newblock {\em Semantic Web journal} 4(4):395--427.

\bibitem[\protect\citeauthoryear{Maier}{2010}]{Maier2010}
Maier, F.
\newblock 2010.
\newblock {Extending Paraconsistent $\mathcal{SROIQ}$}.
\newblock In {\em Web Reasoning and Rule Systems (RR-2010)}, volume 6333 of {\em Lecture Notes in Computer Science}. Berlin, Heidelberg: Springer Berlin Heidelberg.
\newblock  118--132.

\bibitem[\protect\citeauthoryear{Nguyen and Sza{\l}as}{2012}]{NguyenSzalas2012}
Nguyen, L., and Sza{\l}as, A.
\newblock 2012.
\newblock {Paraconsistent Reasoning for Semantic Web Agents}.
\newblock In {\em Proceedings of Transactions on Computational Collective Intelligence VI}, volume 7190 of {\em Lecture Notes in Computer Science}. Berlin, Heidelberg: Springer Berlin Heidelberg.
\newblock  36--55.

\bibitem[\protect\citeauthoryear{Odintsov and Wansing}{2003}]{OdintsovWansing2003}
Odintsov, S., and Wansing, H.
\newblock 2003.
\newblock {Inconsistency-Tolerant Description Logic: Motivation and Basic Systems}.
\newblock {\em Studia Logica} 31:301--335.

\bibitem[\protect\citeauthoryear{Omori and Sano}{2015}]{OmoriSano2015}
Omori, H., and Sano, K.
\newblock 2015.
\newblock {Generalizing Functional Completeness in Belnap-Dunn Logic}.
\newblock {\em Studia Logica} 103(5):883--917.

\bibitem[\protect\citeauthoryear{Ortiz and \v{S}imkus}{2012}]{DBLP:conf/rweb/OrtizS12}
Ortiz, M., and \v{S}imkus, M.
\newblock 2012.
\newblock {Reasoning and Query Answering in Description Logics}.
\newblock In {\em Reasoning Web. Semantic Technologies for Advanced Query Answering}, volume 7487 of {\em Lecture Notes in Computer Science},  1--53.
\newblock Springer.

\bibitem[\protect\citeauthoryear{Poggi \bgroup et al\mbox.\egroup }{2008}]{DBLP:journals/jods/PoggiLCGLR08}
Poggi, A.; Lembo, D.; Calvanese, D.; De~Giacomo, G.; Lenzerini, M.; and Rosati, R.
\newblock 2008.
\newblock {Linking Data to Ontologies}.
\newblock {\em Journal of Data Semantics} 10:133--173.

\bibitem[\protect\citeauthoryear{Sano and Omori}{2014}]{SanoOmori2014}
Sano, K., and Omori, H.
\newblock 2014.
\newblock {An Expansion of First-Order Belnap--Dunn Logic}.
\newblock {\em Logic Journal of IGPL} 22(3):458--481.

\bibitem[\protect\citeauthoryear{Schaerf}{1993}]{Schaerf93}
Schaerf, A.
\newblock 1993.
\newblock {On the Complexity of the Instance Checking Problem in Concept Languages With Existential Quantification}.
\newblock {\em Journal of Intelligent Information Systems} 2(3):265--278.

\bibitem[\protect\citeauthoryear{Skurt}{2020}]{Skurt2020}
Skurt, D.
\newblock 2020.
\newblock {\em {Mehrwertige Aussagenlogik}}.
\newblock Ruhr-Universit\"{a}t Bochum.

\bibitem[\protect\citeauthoryear{Tobies}{2001}]{Tobies01}
Tobies, S.
\newblock 2001.
\newblock {\em {Complexity Results and Practical Algorithms for Logics in Knowledge Representation}}.
\newblock Ph.D. Dissertation, RWTH Aachen University, Germany.

\bibitem[\protect\citeauthoryear{Xiao \bgroup et al\mbox.\egroup }{2018}]{DBLP:conf/ijcai/XiaoCKLPRZ18}
Xiao, G.; Calvanese, D.; Kontchakov, R.; Lembo, D.; Poggi, A.; Rosati, R.; and Zakharyaschev, M.
\newblock 2018.
\newblock {Ontology-Based Data Access: A Survey}.
\newblock In {\em Proceedings of the 27th International Joint Conference on Artificial Intelligence (IJCAI-2018)},  5511--5519.

\bibitem[\protect\citeauthoryear{Zhang \bgroup et al\mbox.\egroup }{2014}]{ZhangXiaoLinvandenBussche2014}
Zhang, X.; Xiao, G.; Lin, Z.; and Bussche, J.
\newblock 2014.
\newblock {Inconsistency-Tolerant Reasoning With OWL DL}.
\newblock {\em International Journal of Approximate Reasoning} 55(2):557--584.

\bibitem[\protect\citeauthoryear{Zhang, Lin, and Wang}{2010}]{ZhangLinWang2010}
Zhang, X.; Lin, Z.; and Wang, K.
\newblock 2010.
\newblock {Towards a Paradoxical Description Logic for the Semantic Web}.
\newblock In {\em Foundations of Information and Knowledge Systems}, volume 5956 of {\em Lecture Notes in Computer Science},  306--325.
\newblock Berlin, Heidelberg: Springer Berlin Heidelberg.

\bibitem[\protect\citeauthoryear{Zhou \bgroup et al\mbox.\egroup }{2012}]{ZhouHuangQiMaHuangQu2012}
Zhou, L.; Huang, H.; Qi, G.; Ma, Y.; Huang, Z.; and Qu, Y.
\newblock 2012.
\newblock {Paraconsistent Query Answering Over DL-Lite Ontologies}.
\newblock {\em Web Intelligence and Agent Systems: An International Journal} 10(1):19--31.

\end{thebibliography}

\newpage
\appendix
\section{Proofs of Section~\ref{sec:ALCfourtriangle}}
Classical counterparts ($\cdot^\cl$) are defined in Definition~\ref{def:triangleembedding}.
\proptriangleembeddingfirst*
\begin{proof}
Let $\Kmc\not\models_\four\phi$ and $\Imc=\langle\Delta^\Imc,\cdot^{\Imc_\pmsf},\cdot^{\Imc_\nmsf}\rangle$ be a~$\four$-interpretation that falsifies the entailment. 
We show that $\Imc^\cl\models\Kmc^\cl$ and $\Imc^\cl\not\models\phi^\cl$, so that $\Kmc^\cl\not\models\phi^\cl$. 
Since role interpretations are not affected by $\cdot^\cl$, it suffices to show that $(\Cmsf^\cl)^{\Imc^\cl}=\Cmsf^{\Imc_\pmsf}$ for every concept $\Cmsf$. We proceed by induction on~$\Cmsf$. The proof 
is mostly the same as that of
~\cite[Proposition~38]{MaierMaHitzler2013}, so we only consider the cases with $\triangle$. The cases of $\Cmsf=\triangle\Amsf$ and $\Cmsf=\triangle\neg\Amsf$ are straightforward since $(\triangle\Dmsf)^{\Imc_\pmsf}=\Dmsf^{\Imc_\pmsf}$ for any concept $\Dmsf$. Let now $\Cmsf=\neg\triangle\Amsf$. Then $(\neg\triangle\Amsf)^{\Imc_\pmsf}=\Delta^\Imc\setminus\Amsf^{\Imc_\pmsf}=\Delta^{\Imc^\cl}\setminus(\Amsf^+)^{\Imc^\cl}=({\neg}\Amsf^+)^{\Imc^\cl}$, as required. The case of $\Cmsf=\neg\triangle\neg\Amsf$ can be tackled similarly: 
$(\neg\triangle\neg\Amsf)^{\Imc_\pmsf}=(\triangle\neg\Amsf)^{\Imc_\nmsf}=\Delta^\Imc\setminus(\neg\Amsf)^{\Imc_\pmsf}=\Delta^\Imc\setminus\Amsf^{\Imc_\nmsf}=\Delta^{\Imc^\cl}\setminus(\Amsf^-)^{\Imc^\cl}=({\neg}\Amsf^-)^{\Imc^\cl}$, as required. 
This also shows that $\Imc\models_\four\Kmc$ iff $\Imc^\cl\models\Kmc^\cl$.

For the converse, let $\Imc=\langle\Delta^\Imc,\cdot^\Imc\rangle$ be a~\emph{classical} interpretation that witnesses $\Kmc^\cl\not\models\phi^\cl$. We define its \emph{$\four$-counterpart} $\Imc^\four=\langle\Delta^\Imc,\cdot^{\Imc^\four_\pmsf},\cdot^{\Imc^\four_\nmsf}\rangle$ as follows: $\Amsf^{\Imc^\four_\pmsf}=(\Amsf^+)^\Imc$ and $\Amsf^{\Imc^\four_\nmsf}=(\Amsf^-)^\Imc$ for $\Amsf^+,\Amsf^-\in\CN$; $\Rmsf^{\Imc^\four}=\Rmsf^\Imc$ for $\Rmsf\in\RN$ and $a^{\Imc^\four}=a^\Imc$ for $a\in\IN$. We show that $\Imc^\four\models_\four\Kmc$ and $\Imc^\four\not\models_\four\phi$ by checking that $\Cmsf^{\Imc^\four_\pmsf}=(\Cmsf^\cl)^\Imc$ for every concept. The proof is again easy to make by induction on $\Cmsf$.
\end{proof}

Recall that $\chi^\triangle$ denotes the result of putting $\triangle$ in front of every concept name occurring in $\chi$.

\lemreductionclassicalparaconsistent*
\begin{proof}
Let $\Imc=\langle\Delta^\Imc,\cdot^\Imc\rangle$ be a~classical model of $\Kmc$. Now, let $\cdot^{\Imc^\four_\nmsf}$ be arbitrary and $\cdot^{\Imc^\four_\pmsf}=\cdot^{\Imc}$ and define $\Imc^\four=\langle\Delta^\Imc,\cdot^{\Imc^\four_\pmsf},\cdot^{\Imc^\four_\nmsf}\rangle$. We show that $\Imc_\four\models_\four\Kmc^\triangle$. 
Since $\cdot^\triangle$ does not affect role names, it suffices to check that $(\Cmsf^\triangle)^{\Imc^\four_\pmsf}=\Cmsf^\Imc$ for every concept. We proceed by induction on $\Cmsf$. The basis case of $\Cmsf=\Amsf$ and $\Cmsf^\triangle=\triangle\Amsf$ follows by the construction of $\Imc_\four$ since $(\triangle\Amsf)^{\Imc^\four_\pmsf}=\Amsf^{\Imc^\four_\pmsf}=\Amsf^\Imc$. If $\Cmsf=\neg\Amsf$, we have that
\begin{align*}
(\neg\Amsf)^\Imc=\Delta^\Imc\setminus\Amsf^\Imc=\Delta^\Imc\setminus\Amsf^{\Imc^\four_\pmsf}=(\triangle\Amsf)^{\Imc^\four_\nmsf}=(\neg\triangle\Amsf)^{\Imc^\four_\pmsf}
\end{align*}
The cases of other connectives and quantifiers follow by an application of the induction hypothesis.

In the other direction, let $\Imc^\four=\langle\Delta^{\Imc^\four},\cdot^{\Imc^\four_\pmsf},\cdot^{\Imc^\four_\nmsf}\rangle$ be a $\four$-model of $\Kmc^\triangle$ and define $\Imc=\langle\Delta^{\Imc^\four},\cdot^\Imc\rangle$ with $\cdot^\Imc=\cdot^{\Imc^\four_\pmsf}$. To show that $\Imc\models\Kmc$, it again suffices to check that $\Cmsf^\Imc=(\Cmsf^\triangle)^{\Imc^\four_\pmsf}$ for every concept, and the proof is similar to the one for the first direction.
\end{proof}

\section{Proofs of Section~\ref{sec:4queries}}

Recall that $\Kmc$ is an $\ALCHItrianglefour$ KB and $\qmbf$ is a BCQV such that the only value operators in $\qmbf$ are $\true$ and $\false$, and that $\Kmc^\flat$ and $\qmbf^\flat$ denote the results of removing all occurrences of $\triangle$ in $\Kmc$ and replacing every $\true(\Amsf(t))$ and $\false(\Amsf(t))$ in $\qmbf$ by $\Amsf(t)$ and $\neg \Amsf(t)$ respectively.

\begin{lemma}\label{lem:four-counterparts}
For every classical (two-valued) model $\Imc=\langle\Delta^\Imc,\cdot^\Imc\rangle$ of $\Kmc^\flat$, the \emph{$\four$-valued counterpart $\Imc^\four$ of $\Imc$} defined by $\Imc^\four=\langle\Delta^\Imc,\cdot^{\Imc^\four_\pmsf},\cdot^{\Imc^\four_\nmsf}\rangle$ with $\Amsf^{\Imc^\four_\pmsf}=\Amsf^\Imc$ and $\Amsf^{\Imc^\four_\nmsf}=\Delta^\Imc\setminus\Amsf^\Imc$ for all $\Amsf\in\CN$, $\Rmsf^{\Imc^\four}=\Rmsf^\Imc$ for all $\Rmsf\in\RN$, and $a^{\Imc^\four}=a^\Imc$ for $a\in\IN$ is such that $\Imc^\four\models_\four\Kmc$. 
\end{lemma}
\begin{proof}
We show that $\Imc^\four\models_\four\Kmc$ by proving by structural induction that for every $\ALCHItrianglefour$ concept $\Cmsf$, $\Cmsf^{\Imc^\four_\pmsf}=(\Cmsf^\flat)^\Imc$ and $\Cmsf^{\Imc^\four_\nmsf}=\Delta^\Imc\setminus(\Cmsf^\flat)^\Imc$. 
\begin{itemize}
\item Base case: $\Amsf^\flat=\Amsf$ and by construction, $\Amsf^{\Imc^\four_\pmsf}=\Amsf^\Imc$ and $\Amsf^{\Imc^\four_\nmsf}=\Delta^\Imc\setminus\Amsf^{\Imc^\four_\pmsf}=\Delta^\Imc\setminus\Amsf^\Imc$.
\item Induction step: let $\Cmsf$ be an $\ALCHItrianglefour$ concept. 

If $\Cmsf=\neg\Dmsf$, then $\Cmsf^{\Imc^\four_\pmsf}=\Dmsf^{\Imc^\four_\nmsf}=\Delta^\Imc\setminus(\Dmsf^\flat)^\Imc=(\neg\Dmsf^\flat)^\Imc=(\Cmsf^\flat)^\Imc$ and $\Cmsf^{\Imc^\four_\nmsf}=\Dmsf^{\Imc^\four_\pmsf}=(\Dmsf^\flat)^\Imc=\Delta^\Imc\setminus(\neg\Dmsf^\flat)^\Imc=\Delta^\Imc\setminus(\Cmsf^\flat)^\Imc$. 

If $\Cmsf=\triangle\Dmsf$, $\Cmsf^{\Imc^\four_\pmsf}=\Dmsf^{\Imc^\four_\pmsf}=(\Dmsf^\flat)^\Imc=(\Cmsf^\flat)^\Imc$ and $\Cmsf^{\Imc^\four_\nmsf}=\Delta^\Imc\setminus\Dmsf^{\Imc^\four_\pmsf}=\Delta^\Imc\setminus(\Dmsf^\flat)^\Imc=\Delta^\Imc\setminus(\Cmsf^\flat)^\Imc$.

Finally, if $\Cmsf=\Dmsf_1\sqcap\Dmsf_2$, $\Cmsf=\Dmsf_1\sqcup\Dmsf_2$, $\Cmsf=\exists\Rmsf.\Dmsf$ or $\Cmsf=\forall\Rmsf.\Dmsf$, observe from~\eqref{equ:ALCHI4propositional} that $\cdot^{\Imc^\four_\pmsf}$ and $\cdot^\Imc$ behave in the same way. We thus obtain $\Cmsf^{\Imc^\four_\pmsf}=(\Cmsf^\flat)^\Imc$ directly from the fact that, e.g., $\Cmsf^\flat=\Dmsf_1^\flat\sqcap\Dmsf_2^\flat$. We show that $\Cmsf^{\Imc^\four_\nmsf}=\Delta^\Imc\setminus(\Cmsf^\flat)^\Imc$ for the case $\Cmsf=\Dmsf_1\sqcap\Dmsf_2$ (the other cases can be shown in the same way). 
$\Cmsf^{\Imc^\four_\nmsf}=\Dmsf_1^{\Imc^\four_\nmsf}\cup\Dmsf_2^{\Imc^\four_\nmsf}=(\Delta^\Imc\setminus(\Dmsf_1^\flat)^\Imc)\cup(\Delta^\Imc\setminus(\Dmsf_2^\flat)^\Imc)=\Delta^\Imc\setminus((\Dmsf_1^\flat)^\Imc\cap(\Dmsf_2^\flat)^\Imc)=\Delta^\Imc\setminus(\Cmsf^\flat)^\Imc$. 
\qedhere
\end{itemize}
\end{proof}

\propsoundness*
\begin{proof}
Assume that $\Kmc^\flat\not\models\qmbf^\flat$. Since $\Kmc^\flat\not\models\qmbf^\flat$, there is a~classical model $\Imc=\langle\Delta^\Imc,\cdot^\Imc\rangle$ of $\Kmc^\flat$ where there is no match $\pi:\Term\mapsto\Delta^\Imc$ s.t.\
\begin{itemize}
\item $\pi(c)=c^\Imc$ for each $c\in\IN$;
\item $(\pi(t),\pi(t'))\in\Rmsf^\Imc$ for each $\Rmsf(t,t')\in\At(\qmbf^\flat)$;
\item $\pi(t)\in\Amsf^\Imc$ for each $\Amsf(t)\in\At(\qmbf^\flat)$;
\item $\pi(t)\notin\Amsf^\Imc$ for each $\neg\Amsf(t)\in\At(\qmbf^\flat)$.
\end{itemize}

Now, let $\Imc^\four=\langle\Delta^\Imc,\cdot^{\Imc^\four_\pmsf},\cdot^{\Imc^\four_\nmsf}\rangle$ be the~$\four$-valued counterpart of $\Imc$ as defined in Lemma~\ref{lem:four-counterparts}. 
Recall that $\Amsf^{\Imc^\four_\pmsf}=\Amsf^\Imc$ and $\Amsf^{\Imc^\four_\nmsf}=\Delta^\Imc\setminus\Amsf^\Imc$ for all $\Amsf\in\CN$. 
Hence, there is no match $\pi:\Term\mapsto\Delta^\Imc$ such that 
\begin{itemize}
\item $\pi(c)=c^{\Imc^\four}$ for each $c\in\IN$;
\item $(\pi(t),\pi(t'))\in\Rmsf^{\Imc^\four}$ for each $\Rmsf(t,t')\in\At(\qmbf)$;
\item $\pi(t)\in\Amsf^{\Imc^\four_\pmsf}$ for each $\Amsf(t)\in\At^+(\qmbf)$;
\item $\pi(t)\in\Amsf^{\Imc^\four_\nmsf}$ for each $\Amsf(t)\in\At^{\both\false}(\qmbf)$.
\end{itemize}
However, since $\Imc^\four\models_\four\Kmc$, 
the existence of such a match is required by item~\ref{item:classicalconditionquery} of Definition~\ref{def:querysanswer} for $\Kmc\models_\four\qmbf$ to hold. Thus, $\Kmc\not\models_\four\qmbf$.
\end{proof}

The following lemma will be used to prove Proposition~\ref{prop:completeness}.

\begin{lemma}\label{lemma:hom-universal-model-horn}
Let $\Kmc=\langle\Tmc,\Amc\rangle$ be a~classically satisfiable $\mathcal{ELHI}_\neg$ KB and $\Jmc_\Kmc=\langle\Delta^{\Jmc_\Kmc},\cdot^{\Jmc_\Kmc}\rangle$ be its two-valued universal model. For every $\four$-model $\Imc=\langle\Delta^\Imc,\cdot^{\Imc_\pmsf},\cdot^{\Imc_\nmsf}\rangle$ of $\Kmc$, there is a~homomorphism between $\Jmc_\Kmc$ and $\langle\Delta^\Imc,\cdot^{\Imc_\pmsf}\rangle$, i.e., there exists $h:\Delta^{\Jmc_\Kmc}\mapsto\Delta^{\Imc}$ such that:
\begin{itemize}
\item $h(a^{\Jmc_\Kmc})=a^\Imc$ for every $a\in\IN$;
\item $(c,d)\in\Rmsf^{\Jmc_\Kmc}$ implies $(h(c),h(d))\in\Rmsf^{\Imc}$;
\item $c\in\Amsf^{\Jmc_\Kmc}$ implies $h(c)\in\Amsf^{\Imc_\pmsf}$.
\end{itemize}
\end{lemma}
\begin{proof}
Recall that $\Jmc_\Kmc=\bigcup_{i\geq 0}\Jmc_i$ where 
$\Delta^{\Jmc_0}=\IN$, 
\begin{itemize}
\item $a^{\Jmc_0}=a$ for every $a\in\IN$;
\item $\Amsf^{\Jmc_0}=\{a\mid\Amsf(a)\in\Amc\}$ for every $\Amsf\in\CN$;
\item $\Rmsf^{\Jmc_0}=\{(a,b)\mid\Rmsf(a,b)\in\Amc\}$ for every $\Rmsf\in\RN$; 
\end{itemize}
and $\Jmc_{i+1}$ results from applying one of the rules below to~$\Jmc_i$ (ignoring TBox inclusions with $\bot$ in the right-hand side).
\begin{enumerate}[noitemsep,topsep=1pt]
\item If $\Smsf_1\!\sqsubseteq\!\Smsf_2\!\in\!\Tmc$ and $(d,\!e)\!\in\!\Smsf_1^{\Jmc_i}$, then $\Smsf_2^{\Jmc_{i+1}}\!=\!\Smsf_2^{\Jmc_i}\!\cup\!\{(d,\!e)\}$.
\item If $\Amsf_1\sqcap \Amsf_2\!\sqsubseteq\!\Bmsf\!\in\!\Tmc$ and $d\!\in\!\Amsf_1^{\Jmc_i}\cap\Amsf_2^{\Jmc_i}$, then $\Bmsf^{\Jmc_{i+1}}\!=\!\Bmsf^{\Jmc_i}\!\cup\!\{d\}$.
\item If $\exists\Smsf.\Amsf\sqsubseteq\Bmsf\in\Tmc$, $(d,e)\in\Smsf^{\Jmc_i}$, and $e\in\Amsf^{\Jmc_i}$, then $\Bmsf^{\Jmc_{i+1}}=\Bmsf^{\Jmc_i}\cup\{d\}$.
\item If $\Amsf\sqsubseteq\exists\Smsf.\Bmsf\in\Tmc$ and $d\in\Amsf^{\Jmc_i}$, then $\Delta^{\Jmc_{i+1}}=\Delta^{\Jmc_i}\cup\{x\}$, $\Smsf^{\Jmc_{i+1}}=\Smsf^{\Jmc_i}\cup\{(d,x)\}$ and $\Bmsf^{\Jmc_{i+1}}=\Bmsf^{\Jmc_i}\cup\{x\}$ where $x\notin\Delta^{\Jmc_i}$ is a~fresh domain element.
\end{enumerate}

Let $\Imc=\langle\Delta^\Imc,\cdot^{\Imc_\pmsf},\cdot^{\Imc_\nmsf}\rangle$ be a~$\four$-model of $\Kmc$. It is easy to build inductively a~homomorphism $h\coloneqq\bigcup\limits_{i\geq 0}h_i$ as required. We start with $h_0$ being the identity on $\IN$, which is such that 
\begin{itemize}
\item $h_0(a^{\Jmc_0})=a^\Imc$ for every $a\in\IN$;
\item $(c,d)\in\Rmsf^{\Jmc_0}$ implies $\Rmsf(c,d)\in\Amc$, so $(h_0(c),h_0(d))\in\Rmsf^{\Imc}$;
\item $c\in\Amsf^{\Jmc_0}$ implies $\Amsf(c)\in\Amc$ so $h_0(c)\in\Amsf^{\Imc_\pmsf}$.
\end{itemize}

Then for every $i\geq 0$, if $h_i:\Delta^{\Jmc_i}\mapsto\Delta^{\Imc}$ is a~homomorphism from $\Jmc_i$ to $\langle\Delta^\Imc,\cdot^{\Imc_\pmsf}\rangle$, we build a~homomorphism $h_{i+1}:\Delta^{\Jmc_{i+1}}\mapsto\Delta^{\Imc}$ from $\Jmc_{i+1}$ to $\langle\Delta^\Imc,\cdot^{\Imc_\pmsf}\rangle$ as follows.

If $\Jmc_{i+1}$ has been obtained from $\Jmc_i$ by applying rule 1, 2 or 3, let $h_i=h_{i+1}$. Since $\Imc$ is a~model of $\Tmc$, it is easy to check that $h_{i+1}$ is still a~homomorphism. Otherwise, if $\Jmc_{i+1}$ has been obtained from $\Jmc_i$ by applying rule 4 to $\Amsf\sqsubseteq\exists\Smsf.\Bmsf\in\Tmc$ and $d\in\Amsf^{\Jmc_i}$, by induction we obtain $h_i(d)\in\Amsf^{\Imc_\pmsf}$. So, since $\Imc\models_\four\Amsf\sqsubseteq\exists\Smsf.\Bmsf$, there must exists $e$ s.t.\ $(h_i(d),e)\!\in\!\Smsf^\Imc$ and $e\!\in\!\Bmsf^{\Imc_\pmsf}$ and we define $h_{i+1}(x)=e$, so that $(h_{i+1}(d),h_{i+1}(x))\!\in\!\Smsf^\Imc$ and $h_{i+1}(x)\!\in\!\Bmsf^{\Imc_\pmsf}$.
\end{proof}

\propcompleteness*
\begin{proof}
Assume that $\Kmc^\flat\models\qmbf^\flat$ and assume for a contradiction that $\Kmc\not\models_\four \qmbf$, i.e., that one of the following conditions holds.
\begin{enumerate}[noitemsep,topsep=1pt]
\item[(i)] There exists a $\four$-model $\Imc=\langle\Delta^\Imc,\cdot^{\Imc_\pmsf},\cdot^{\Imc_\nmsf}\rangle$ of $\Kmc$ such that there is no match $\pi:\Term\mapsto \Delta^\Imc$ as required by item~\ref{item:classicalconditionquery} of Definition \ref{def:querysanswer}.
\item[(ii)] For every $\four$-model $\Imc=\langle\Delta^\Imc,\cdot^{\Imc_\pmsf},\cdot^{\Imc_\nmsf}\rangle$ of $\Kmc$ and every match $\pi$ as required by item~\ref{item:classicalconditionquery} of Definition~\ref{def:querysanswer}, there exists $\true(\Amsf(t))\in\At(\qmbf)$ such that $\pi(t)\in \Amsf^{\Imc_\nmsf}$. 
\end{enumerate}

In case (i), first note that since $\ELHInegtrianglefour$ KBs actually do not contain $\triangle$, $\Kmc=\Kmc^\flat$ so $\Imc\models_\four\Kmc$ implies that $\Imc\models_\four\Kmc^\flat$. 
Hence, by Lemma \ref{lemma:hom-universal-model-horn}, there is a homomorphism $h:\Delta^{\Jmc_\Kmc}\mapsto\Delta^{\Imc}$ from the two-valued universal model $\Jmc_\Kmc$ of $\Kmc^\flat$ and $\langle\Delta^\Imc,\cdot^{\Imc_\pmsf}\rangle$ such that:
 \begin{itemize}
 \item $h(a^{\Jmc_\Kmc})=a^\Imc$ for every $a\in\IN$,
 \item $(c,d)\in \Rmsf^{\Jmc_\Kmc}$ implies $(h(c),h(d))\in \Rmsf^{\Imc}$,
 \item $c\in \Amsf^{\Jmc_\Kmc}$ implies $h(c)\in \Amsf^{\Imc_\pmsf}$.
 \end{itemize}
 Since $\Jmc_\Kmc\models \Kmc^\flat$, $\Kmc^\flat\models \qmbf^\flat$ implies that $\Jmc_\Kmc\models \qmbf^\flat$ so there is a match $\pi$ for $\qmbf^\flat$ in $\Jmc_\Kmc$. We show that $\pi'=h\circ \pi$ is a match for $\qmbf$ in $\Imc$ as required by item~\ref{item:classicalconditionquery} of Definition \ref{def:querysanswer}, contradicting~(i). 
 For every $c\in\IN$, $\pi'(c)=h(\pi(c))=h(c^{\Jmc_\Kmc})=c^\Imc$, and 
 \begin{itemize}
 \item for every $\Rmsf(t_1,t_2)\in\At(\qmbf)$, $(\pi(t_1),\pi(t_2))\in \Rmsf^{\Jmc_\Kmc}$ so $(h(\pi(t_1)),h(\pi(t_2)))=(\pi'(t_1),\pi'(t_2))\in \Rmsf^\Imc$,
 \item for every $\Amsf(t)\in\At(\qmbf)$, $\Amsf(t)\in\At(\qmbf^\flat)$ so $\pi(t)\in \Amsf^{\Jmc_\Kmc}$ and $h(\pi(t))=\pi'(t)\in \Amsf^{\Imc_\pmsf}$,
 \item for every $\true(\Amsf(t))\in\At(\qmbf)$, $\Amsf(t)\in\At(\qmbf^\flat)$ so as above $\pi'(t)\in \Amsf^{\Imc_\pmsf}$. 
 \end{itemize} 

In case (ii), let $\Jmc$ be a two-valued model of $\Kmc^\flat$ and 
let $\Jmc^\four=\langle\Delta^\Jmc,\cdot^{\Jmc^\four_\pmsf},\cdot^{\Jmc^\four_\nmsf}\rangle$ be the~$\four$-valued counterpart of $\Jmc$ as defined in Lemma~\ref{lem:four-counterparts}. 
 By construction, for every $\Amsf\in\CN$, there is no $e\in\Delta^{\Jmc}$ such that $e\in \Amsf^{\Jmc^\four_\pmsf}\cap \Amsf^{\Jmc^\four_\nmsf}$. In particular, for every match $\pi$ of $\qmbf$ in $\Jmc^\four$ as required by item~\ref{item:classicalconditionquery} of Definition \ref{def:querysanswer}, for every $\true(\Amsf(t))\in\At(\qmbf)$, since $\pi(t)\in \Amsf^{\Jmc^\four_\pmsf}$ is required by the definition of $\pi$, it follows that $\pi(t)\notin \Amsf^{\Jmc^\four_\nmsf}$. 
 Hence, $\Jmc^\four$ is a $\four$-model of $\Kmc$ that contradicts (ii).

We obtain a contradiction in both cases so we conclude that $\Kmc\models_\four \qmbf$.
\end{proof}

\section{Proofs of Section~\ref{sec:complexity}}

\subsection{Proof of Theorem \ref{theorem:queryreduction}}
Recall that $\qmbf=\exists\vec{y}:\varphi$ is a~BCQV, $\IN_\qmbf=\{c_x\mid x\in\Var\text{ occurs in }\qmbf\}$, 
\begin{align*}
c_t&=\begin{cases}t\text{ if }t\in\IN\\
c_t\text{ if }t\in\Var
\end{cases}
\end{align*}
and 
\begin{align*}
\qmbf^+&\coloneqq\exists\vec{y}:\bigwedge\limits_{\Rmsf(t,t')\in\At(\qmbf)}\!\!\!\!\!\!\Rmsf(t,t')\wedge\!\!\!\!\!\!\bigwedge\limits_{\Amsf(t)\in\At^+(\qmbf)}\!\!\!\!\!\!\!\!\!\Amsf^+(t)\wedge\!\!\!\!\!\!\bigwedge\limits_{\Amsf(t)\in\At^{\both\false}(\qmbf)}\!\!\!\!\!\!\!\!\!\Amsf^-(t)\\
\qmbf^\ctr&\coloneqq\bigvee\limits_{\Amsf(t)\in\At^{\true\neither}(\qmbf)}\Amsf^-(c_t)\vee\bigvee\limits_{\Amsf(t)\in\At^{\false\neither}(\qmbf)}\Amsf^+(c_t)\\
\Amc_\qmbf&\coloneqq\{\Rmsf(c_t,c_{t'})\mid\Rmsf(t,t')\in\At(\qmbf^+)\}\cup\\&\hspace{1.5em}\{\Amsf^+(c_t)\mid\Amsf^+(t)\in\At(\qmbf^+)\}\cup\\&\hspace{1.5em}\{\Amsf^-(c_t)\mid\Amsf^-(t)\in\At(\qmbf^+)\}
\end{align*}

\theoremqueryreduction*
\begin{proof}
\noindent($\Rightarrow$) Assume that $\Kmc\models_\four\qmbf$.

Let $\Imc^\cl=\langle\Delta^{\Imc^\cl},\cdot^{\Imc^\cl}\rangle$ be a~(classical) model of $\Kmc^\cl$ and $(\Imc^\cl)^\four=\langle\Delta^{(\Imc^\cl)^\four},\cdot^{(\Imc^\cl)^\four_\pmsf},\cdot^{(\Imc^\cl)^\four_\nmsf}\rangle$ be its \emph{$\four$-counterpart} defined as in the proof of Proposition~\ref{prop:triangleembedding1}: $(\Imc^\cl)^\four\models_\four\Kmc$. Thus, by item~\ref{item:classicalconditionquery} of Definition~\ref{def:querysanswer} there is a~match $\pi:\Term\mapsto\Delta^{(\Imc^\cl)^\four}=\Delta^{\Imc^\cl}$ such that $\pi(c)=c^{(\Imc^\cl)^\four}=c^{\Imc^\cl}$ for every $c\in\IN$, and
\begin{itemize}
\item $(\pi(t_1),\pi(t_2))\in\Rmsf^{(\Imc^\cl)^\four}=\Rmsf^{\Imc^\cl}$ for each $\Rmsf(t_1,t_2)\in\At(\qmbf)$;
\item $\pi(t)\in\Amsf^{(\Imc^\cl)^\four_\pmsf}=(\Amsf^+)^{\Imc^\cl}$ for each $\Amsf(t)\in\At^+(\qmbf)$; and
\item $\pi(t)\in \Amsf^{(\Imc^\cl)^\four_\nmsf}=(\Amsf^-)^{\Imc^\cl}$ for each $\Amsf(t)\in\At^{\both\false}(\qmbf)$.
\end{itemize}
Hence, by definition of $\qmbf^+$, $\pi$ is a~match for $\qmbf^+$ in $\Imc^\cl$. It follows that $\Kmc^\cl\models\qmbf^+$.

Moreover, by item~\ref{item:valueconditionquery} of Definition~\ref{def:querysanswer}, there exists a~$\four$-model $\Imc'=\langle\Delta^{\Imc'},\cdot^{\Imc'_\pmsf},\cdot^{\Imc'_\nmsf}\rangle$ of $\Kmc$ and a~match $\pi:\Term\mapsto\Delta^{\Imc'}$ such that $\pi(c)=c^{\Imc'}$ for every $c\in\IN$, and
\begin{itemize}
\item $(\pi(t_1),\pi(t_2))\in\Rmsf^{\Imc'}$ for every $\Rmsf(t_1,t_2)\in\At(\qmbf)$;
\item $\pi(t)\in\Amsf^{\Imc'_\pmsf}$ for every $\Amsf(t)\in\At(\qmbf)$;
\item $\pi(t)\in\Amsf^{\Imc'_\pmsf}$ and $\pi(t)\notin\Amsf^{\Imc'_\nmsf}$ for every $\true(\Amsf(t))\in\At(\qmbf)$;
\item $\pi(t)\in\Amsf^{\Imc'_\nmsf}$ and $\pi(t)\notin\Amsf^{\Imc'_\pmsf}$ for every $\false(\Amsf(t))\in\At(\qmbf)$;
\item $\pi(t)\in\Amsf^{\Imc'_\pmsf}\cap\Amsf^{\Imc'_\nmsf}$ for every $\both(\Amsf(t))\in\At(\qmbf)$;
\item $\pi(t)\notin\Amsf^{\Imc'_\pmsf}\cup\Amsf^{\Imc'_\nmsf}$ for every $\neither(\Amsf(t))\in\At(\qmbf)$.
\end{itemize}

Let $\Imc'^\cl$ be the classical counterpart of $\Imc'$ (cf.\ Definition~\ref{def:triangleembedding}). By Proposition~\ref{prop:triangleembedding1}, $\Imc'^\cl\models\Kmc^\cl$. Let $\Jmc$ be the classical interpretation defined by extending $\cdot^{\Imc'^\cl}$ to the individuals from $\IN_\qmbf$ as follows: $c_t^\Jmc=\pi(t)$ for every $c_t\in\IN_\qmbf$. 
We show that $\Jmc\models\Kmc^\cl\cup\Amc_\qmbf$. Since $\Imc'^\cl\models\Kmc^\cl$, we have $\Jmc\models\Kmc^\cl$. Moreover: 
\begin{itemize}
\item for every $\Rmsf(c_t,c_{t'})\in\Amc_\qmbf$ since $\Rmsf(t,t')\in\At(\qmbf^+)$, i.e., $\Rmsf(t,t')\in\At(\qmbf)$ and $\pi$ is a~match for $\qmbf$ in $\Imc'$, $(c_t^\Jmc,c_{t'}^\Jmc)=(\pi(t),\pi(t'))\in\Rmsf^{\Imc'}=\Rmsf^{\Imc'^\cl}=\Rmsf^\Jmc$;
\item for every $\Amsf^+(c_t)\in\Amc_\qmbf$ since $\Amsf^+(t)\in\At(\qmbf^+)$, then either $\Amsf(t)$, $\true(\Amsf(t))$, or $\both(\Amsf(t))$ is in $\qmbf$, and since $\pi$ is a~match for $\qmbf$ in $\Imc'$, $c_t^\Jmc=\pi(t)\in\Amsf^{\Imc'_\pmsf}$, so $c_t^\Jmc\in(\Amsf^+)^{\Imc'^\cl}=(\Amsf^+)^{\Jmc}$ ; and
\item for every $\Amsf^-(c_t)\in\Amc_\qmbf$, since $\Amsf^-(t)\in\At(\qmbf^+)$, then either $\false(\Amsf(t))$ or $\both(\Amsf(t))$ is in $\qmbf$, and since $\pi$ is a~match for $\qmbf$ in $\Imc'$, $c_t^\Jmc=\pi(t)\in\Amsf^{\Imc'_\nmsf}$, so $c_t^\Jmc\in(\Amsf^-)^{\Imc'^\cl}=(\Amsf^-)^{\Jmc}$.
\end{itemize} 
Thus, $\Jmc\models\Amc_\qmbf$, whence $\Jmc\models\Kmc^\cl\cup\Amc_\qmbf$. In addition, for every term $t$ of $\qmbf$, we have that $c_t^\Jmc=\pi(t)$ (by definition of $c_t^\Jmc$ if $c_t\in\IN_\qmbf$ and otherwise, if $c_t=t\in\IN$, $t^\Jmc=t^{\Imc'^\cl}=t^{\Imc'}=\pi(t)$). Hence,
\begin{itemize}
\item for each $\true(\Amsf(t))\in\At(\qmbf)$ since by definition of $\Imc'$, $\pi(t)\notin\Amsf^{\Imc'_\nmsf}$, then $c_t^\Jmc=\pi(t)\notin(\Amsf^-)^{\Imc'^\cl}=(\Amsf^-)^\Jmc$, i.e., $\Jmc\not\models\Amsf^-(c_t)$;
\item for each $\false(\Amsf(t))\in\At(\qmbf)$ since by definition of $\Imc'$, $\pi(t)\notin\Amsf^{\Imc'_\pmsf}$, then $c_t^\Jmc=\pi(t)\notin(\Amsf^+)^{\Imc'^\cl}=(\Amsf^+)^\Jmc$, i.e., $\Jmc\not\models\Amsf^+(c_t)$;
\item for every $\neither(\Amsf(t))\in\At(\qmbf)$, since by definition of $\Imc'$, $\pi(t)\notin\Amsf^{\Imc'_\pmsf}\cup\Amsf^{\Imc'_\nmsf}$, $\Jmc\not\models\Amsf^+(c_t)$ and $\Jmc\not\models\Amsf^-(c_t)$.
\end{itemize}
It follows that $\Jmc\not\models\qmbf_\ctr$, thus $\Kmc^\cl\cup\Amc_\qmbf\not\models\qmbf_\ctr$.

\noindent($\Leftarrow$) For the converse direction, assume that $\Kmc^\cl\models\qmbf^+$ and $\Kmc^\cl\cup\Amc_\qmbf\not\models\qmbf_{\ctr}$. Note that there exist $\four$-models of $\Kmc$ because $\Kmc^\cl$ has classical models (otherwise it would hold that $\Kmc^\cl\cup\Amc_\qmbf\models\qmbf_{\ctr}$) and the $\four$-counterparts of these models as defined in the proof of Proposition~\ref{prop:triangleembedding1} are $\four$-models of $\Kmc$. For every $\Imc\models_\four\Kmc$, by Proposition~\ref{prop:triangleembedding1}, the classical counterpart $\Imc^\cl$ of $\Imc$ is a~model of $\Kmc^\cl$. Thus, since $\Kmc^\cl\models\qmbf^+$, $\Imc^\cl\models\qmbf^+$. It follows that there is a~match $\pi$ for $\qmbf^+$ in $\Imc^\cl$. Hence $\pi:\Term\mapsto \Delta^{\Imc^\cl}=\Delta^\Imc$ is such that $\pi(c)=c^{\Imc^\cl}=c^\Imc$ for all $c\in\IN$, and
\begin{itemize}
\item $(\pi(t_1),\pi(t_2))\in\Rmsf^{\Imc^\cl}=\Rmsf^\Imc$ for each $\Rmsf(t_1,t_2)\in\At(\qmbf)$;
\item $\pi(t)\in(\Amsf^+)^{\Imc^\cl}=\Amsf^{\Imc_\pmsf}$ for every $\Amsf(t)\in\At^+(\qmbf)$; and
\item $\pi(t)\in(\Amsf^-)^{\Imc^\cl}=\Amsf^{\Imc_\nmsf}$ for every $\Amsf(t)\in\At^{\both\false}(\qmbf)$.
\end{itemize}
It follows that $\pi$ is a~match for $\qmbf$ in $\Imc$ as required by item~\ref{item:classicalconditionquery} of Definition~\ref{def:querysanswer}.

Moreover, since $\Kmc^\cl\cup\Amc_\qmbf\not\models\qmbf_{\ctr}$, there is a~(classical) model $\Jmc=\langle\Delta^\Jmc,\cdot^\Jmc\rangle$ of $\Kmc^\cl\cup\Amc_\qmbf$ such that $\Jmc\not\models\qmbf_\ctr$. Let now $\Jmc^\four=\langle\Delta^{\Jmc^\four},\cdot^{\Jmc^\four_\pmsf},\cdot^{\Jmc^\four_\nmsf}\rangle$ be the $\four$-counterpart of~$\Jmc$ as defined in the proof of Proposition~\ref{prop:triangleembedding1}. Since $\Jmc\models\Kmc^\cl$, we have shown in the proof of Proposition~\ref{prop:triangleembedding1} that $\Jmc^\four\models_\four\Kmc$.

Define $\pi\!:\!\Term\cup\IN_\qmbf\!\mapsto\!\Delta^{\Jmc^\four}$ by $\pi(t)\!=\!c_t^\Jmc$ for every term $t$ of $\qmbf$. We show that $\pi$ is a~match for $\qmbf$ in $\Jmc^\four$ as required by item~\ref{item:valueconditionquery} of Definition~\ref{def:querysanswer}.
\begin{itemize}
\item For every $c\in\IN$, $\pi(c)=c^\Jmc=c^{\Jmc^\four}$.
\item For each $\Rmsf(t_1,t_2)\in\At(\qmbf)$, since $\Jmc\models\Amc_\qmbf$ and $\Rmsf(c_{t_1},c_{t_2})\in\Amc_\qmbf$, then $(\pi(t_1),\pi(t_2))=(c_{t_1}^\Jmc,c_{t_2}^\Jmc)\in\Rmsf^\Jmc=\Rmsf^{\Jmc^\four}$ .
\item For each $\Amsf(t)\in\At(\qmbf)$, $\Amsf^+(t)\!\in\!\At(\qmbf^+)$, so $\Amsf^+(c_t)\!\in\!\Amc_\qmbf$ and since $\Jmc\models\Amc_\qmbf$, then $\pi(t)\!=\!c_t^\Jmc\!\in\!(\Amsf^+)^\Jmc=\Amsf^{\Jmc^\four_\pmsf}$.
\item For every $\true(\Amsf(t))\in\At(\qmbf)$, $\Amsf^+(t)\in\At(\qmbf^+)$. So, we obtain as above that $\pi(t)\in\Amsf^{\Jmc^\four_\pmsf}$. Moreover, since $\Jmc\not\models\qmbf_\ctr$, and in particular $\Jmc\not\models\Amsf^-(c_t)$, we have $\pi(t)=c_t^\Jmc\notin(\Amsf^-)^{\Jmc}$, i.e., $\pi(t)\notin\Amsf^{\Jmc^\four_\nmsf}$.
\item For every $\false(\Amsf(t))\in\At(\qmbf)$, $\Amsf^-(t)\in\At(\qmbf^+)$, so $\Amsf^-(c_t)\in\Amc_\qmbf$. Since $\Jmc\models\Amc_\qmbf$, we thus have $\pi(t)=c_t^\Jmc\in(\Amsf^-)^\Jmc=\Amsf^{\Jmc^\four_\nmsf}$. Moreover, since $\Jmc\not\models\qmbf_\ctr$, and in particular $\Jmc\not\models\Amsf^+(c_t)$, we have $\pi(t)=c_t^\Jmc\notin (\Amsf^+)^{\Jmc}$, i.e., $\pi(t)\notin\Amsf^{\Jmc^\four_\pmsf}$.
\item For every $\both(\Amsf(t))\in\At(\qmbf)$, $\{\Amsf^+(t),\Amsf^-(t)\}\subseteq\At(\qmbf^+)$. So, we obtain as before that $\pi(t)\in(\Amsf^+)^\Jmc\cap(\Amsf^-)^\Jmc$, i.e., $\pi(t)\in\Amsf^{\Jmc^\four_\pmsf}\cap\Amsf^{\Jmc^\four_\nmsf}$.
\item For every $\neither(\Amsf(t))\in\At(\qmbf)$, since $\Jmc\not\models\qmbf_\ctr$, and in particular $\Jmc\not\models\Amsf^-(c_t)$ and $\Jmc\not\models\Amsf^+(c_t)$, $\pi(t)=c_t^\Jmc\notin(\Amsf^+)^\Jmc\cup(\Amsf^-)^\Jmc$, i.e., $\pi(t)\notin\Amsf^{\Jmc^\four_\pmsf}\cup\Amsf^{\Jmc^\four_\nmsf}$.
\end{itemize}
Hence, $\Kmc\models_\four\qmbf$.
\end{proof}

\subsection{Missing Lower Bound for Theorem \ref{complexity}}
\begin{proposition}
Boolean CQV entailment over $\mathcal{ALC}^\four_\triangle$ is $\bhtwo$-hard w.r.t.\ data complexity.
\end{proposition}
\begin{proof}
We use a reduction from the $\bhtwo$-complete problem SAT-UNSAT: given a pair $(\varphi_1,\varphi_2)$ of propositional formulas, decide whether $\varphi_1$ is satisfiable and $\varphi_2$ is unsatisfiable. We assume that $\varphi_1$ and $\varphi_2$ are sets of clauses such that each clause has exactly 2 positive and 2 negative literals, and any of the four positions in a clause can be filled instead by one of the truth constants true and false (this is w.l.o.g.\ since 2+2SAT is $\np$-complete~\cite{Schaerf93}). 
Let $\varphi_1$ and $\varphi_2$ be two set of clauses $\{c^1_1,\dots, c^1_{m_1}\}$ and $\{c^2_1,\dots, c^2_{m_2}\}$ respectively, over variables $x^1_1,\dots,x^1_{n_1}$ and $x^2_1,\dots,x^2_{n_2}$ respectively. 
We define an $\mathcal{ALC}^\four_\triangle$ knowledge base $\Kmc=\langle\Tmc,\Amc\rangle$ and a BCQV $\qmbf$ as follows. 
\begin{align*}
 \Amc=&\{\mathsf{Var}(x^i_j)\mid 1\leq i\leq 2, 1\leq j\leq n_i\}\cup\{\mathsf{F}(f),\mathsf{T}(t)\}\\
 &\{\mathsf{Sat}(\varphi_1)\}\cup\{\mathsf{Cl}(\varphi_i,c^i_k)\mid 1\leq i\leq 2, 1\leq k\leq m_i\}\cup \\
 &\{\mathsf{P}_\ell(c^i_k,x^i_j) \mid \ell\in\{1, 2\},x^i_j\text{ is the }\ell^{th}\text{ pos. lit. of }c^i_k \}\cup\\
 &\{\mathsf{N}_\ell(c^i_k,x^i_j) \mid \ell\in\{1, 2\},\neg x^i_j\text{ is the }\ell^{th}\text{ neg. lit. of }c^i_k \}
 \\
 \Tmc=&\{\mathsf{Var}\sqsubseteq \mathsf{T}\sqcup\mathsf{F}, \mathsf{T}\sqsubseteq \neg\mathsf{F}, \mathsf{F}\sqsubseteq \neg\mathsf{T}\}\cup\\&
 \{ \exists\mathsf{P_1}.\mathsf{F}\sqcap\exists\mathsf{P_2}.\mathsf{F}\sqcap \exists\mathsf{N_1}.\mathsf{T}\sqcap \exists\mathsf{N_2}.\mathsf{T}\sqsubseteq \mathsf{NotSat} \}\cup\\
 &\{\exists\mathsf{Cl}.\mathsf{NotSat}\sqsubseteq\mathsf{NotSat}, \mathsf{NotSat}\sqsubseteq \neg\mathsf{Sat}\}\\
 \qmbf=&\true(\mathsf{Sat}(\varphi_1))\wedge \mathsf{NotSat}(\varphi_2)
\end{align*}
We show that $\Kmc\models_\four \qmbf$ iff $\varphi_1$ is satisfiable and $\varphi_2$ is unsatisfiable. 

\noindent ($\Rightarrow$) Assume that $\Kmc\models_\four \qmbf$. 

Since $\Kmc\models_\four\true(\mathsf{Sat}(\varphi_1))$, by item~\ref{item:valueconditionquery} of Definition~\ref{def:querysanswer}, there exists a $\four$-model $\Imc=\langle\Delta^\Imc,\cdot^{\Imc_\pmsf},\cdot^{\Imc_\nmsf}\rangle$ of $\Kmc$ such that $\varphi_1^\Imc\notin \mathsf{Sat}^{\Imc_\nmsf}$. Since $\Imc\models_\four \mathsf{NotSat}\sqsubseteq \neg\mathsf{Sat}$, it follows that $\varphi_1^\Imc\notin \mathsf{NotSat}^{\Imc_\pmsf}$. 
Hence, for every $1\leq k\leq m_1$, ${c^1_k}^\Imc\notin \mathsf{NotSat}^{\Imc_\pmsf}$. 
It follows that for every $1\leq k\leq m_1$, at least one of the following conditions holds: either (i) $x^1_j$ is a positive literal of $c^1_k$ and ${x^1_j}^\Imc\notin \mathsf{F}^{\Imc_\pmsf}$, so ${x^1_j}^\Imc\in \mathsf{T}^{\Imc_\pmsf}$ (since either ${x^1_j}^\Imc\in \mathsf{Var}^{\Imc_\pmsf}$ or $x^1_j=t$) or (ii) $\neg x^1_j$ is a negative literal of $c^1_k$ and ${x^1_j}^\Imc\notin \mathsf{T}^{\Imc_\pmsf}$, so ${x^1_j}^\Imc\in \mathsf{F}^{\Imc_\pmsf}$ (since either ${x^1_j}^\Imc\in \mathsf{Var}^{\Imc_\pmsf}$ or $x^1_j=f$). 
Let $\nu$ be the valuation of $x^1_1,\dots,x^1_{n_1}$ defined by $\nu(x^1_j)=\text{true}$ iff ${x^1_j}^\Imc\in \mathsf{T}^{\Imc_\pmsf}$. 
For every $1\leq k\leq m_1$, if we are in case (i), there is a positive literal $x^1_j$ of $c^1_k$ such that $\nu(x^1_j)=\text{true}$, and if we are in case (ii), there is a negative literal $\neg x^1_j$ of $c^1_k$ such that $\nu(x^1_j)=\text{false}$. Hence $\nu$ satisfies all clauses of $\varphi_1$ and $\varphi_1$ is satisfiable.

Let $\nu$ be a valuation of $x^2_1,\dots,x^2_{n_2}$ and let $\Imc^\nu$ be the $\four$-interpretation defined as follows: $\Delta^{\Imc^\nu}$ is the set of individuals that occur in $\Amc$, $\Amsf^{\Imc^\nu_\nmsf}=\Delta^{\Imc^\nu}$ for every $\Amsf\in\CN$, $\mathsf{Var}^{\Imc^\nu_\pmsf}$, $\mathsf{P}_\ell^{\Imc^\nu}$, $\mathsf{N}_\ell^{\Imc^\nu}$, $\mathsf{Cl}^{\Imc^\nu}$ and $\mathsf{Sat}^{\Imc^\nu_\pmsf}$ correspond exactly to the assertions in $\Amc$, 
$\mathsf{T}^{\Imc^\nu_\pmsf}=\{t\}\cup\{x^2_j\mid \nu(x^2_j)=\text{true}\}\cup\{x^1_1,\dots,x^1_{n_1}\}$, $\mathsf{F}^{\Imc^\nu_\pmsf}=\{f\}\cup\{x^2_j\mid \nu(x^2_j)=\text{false}\}\cup\{x^1_1,\dots,x^1_{n_1}\}$, and $\mathsf{NotSat}^{\Imc^\nu_\pmsf}=\{c^2_k\mid \nu(c^2_k)=\text{false}\}\cup\{\varphi_2\mid \text{if }\nu(\varphi_2)=\text{false}\}\cup\{\varphi_1, c^1_1,\dots, c^1_{m_1}\}$. 
It is easy to check that $\Imc^\nu\models_\four\Kmc$. 
Hence, since $\Kmc\models_\four \mathsf{NotSat}(\varphi_2)$, by item~\ref{item:classicalconditionquery} of Definition~\ref{def:querysanswer}, it must be the case that $\varphi_2^{\Imc^\nu}\in \mathsf{NotSat}^{\Imc^\nu_\pmsf}$, i.e., $\nu$ falsifies $\varphi_2$. Since this is true for every valuation $\nu$, $\varphi_2$ is unsatisfiable. 

\noindent ($\Leftarrow$) Assume that $\varphi_1$ is satisfiable and $\varphi_2$ is unsatisfiable.

Since $\varphi_1$ is satisfiable, there exists a valuation $\nu$ of $x^1_1,\dots,x^1_{n_1}$ that satisfies it. 
Let $\Imc^\nu$ be the $\four$-interpretation defined as follows: $\Delta^{\Imc^\nu}$ is the set of individuals that occur in $\Amc$, $\Amsf^{\Imc^\nu_\nmsf}=\Delta^{\Imc^\nu}$ for every $\Amsf\in\CN$ except $\mathsf{Sat}$, for which we define $\mathsf{Sat}^{\Imc^\nu_\nmsf}=\{\varphi_2, c^2_1,\dots, c^2_{m_2}\}$, $\mathsf{Var}^{\Imc^\nu_\pmsf}$, $\mathsf{P}_\ell^{\Imc^\nu}$, $\mathsf{N}_\ell^{\Imc^\nu}$, $\mathsf{Cl}^{\Imc^\nu}$ and $\mathsf{Sat}^{\Imc^\nu_\pmsf}$ correspond exactly to the assertions in $\Amc$, 
$\mathsf{T}^{\Imc^\nu_\pmsf}=\{t\}\cup\{x^1_j\mid \nu(x^1_j)=\text{true}\}\cup\{x^2_1,\dots,x^2_{n_2}\}$, $\mathsf{F}^{\Imc^\nu_\pmsf}=\{f\}\cup\{x^1_j\mid \nu(x^1_j)=\text{false}\}\cup\{x^2_1,\dots,x^2_{n_1}\}$, and $\mathsf{NotSat}^{\Imc^\nu_\pmsf}=\{\varphi_2, c^2_1,\dots, c^2_{m_2}\}$. 
It is easy to check that $\Imc^\nu\models_\four\Kmc$. In particular, since $\nu$ satisfies every clause of $\varphi_1$, there is no $c^1_k$ such that ${c^1_k}^{\Imc^\nu_\pmsf}\in (\exists\mathsf{P_1}.\mathsf{F}\sqcap\exists\mathsf{P_2}.\mathsf{F}\sqcap \exists\mathsf{N_1}.\mathsf{T}\sqcap \exists\mathsf{N_2}.\mathsf{T})^{\Imc^\nu_\pmsf}$. It is thus a $\four$-model of $\Kmc$ such that $\varphi_1^{\Imc^\nu}\notin \mathsf{Sat}^{\Imc^\nu_\nmsf}$. 
Since every $\four$-model $\Imc$ of $\Kmc$ is such that $\varphi_1^{\Imc}\in \mathsf{Sat}^{\Imc_\pmsf}$ because $\mathsf{Sat}(\varphi_1)\in\Amc$, it follows by Definition~\ref{def:querysanswer} that $\Kmc\models_\four \true(\mathsf{Sat}(\varphi_1))$. 

Let $\Imc=\langle\Delta^\Imc,\cdot^{\Imc_\pmsf},\cdot^{\Imc_\nmsf}\rangle$ be a $\four$-model of $\Kmc$ and let $\nu$ be the valuation of $x^2_1,\dots,x^2_{n_2}$ defined by $\nu(x^2_j)=\text{true}$ iff ${x^2_j}^\Imc\in \mathsf{T}^{\Imc_\pmsf}$. Since $\varphi_2$ is unsatisfiable, there exists $c^2_k$ such that $\nu(c^2_k)=\text{false}$. It follows that for every positive literal $x^2_j$ of $c^2_k$, $\nu(x^2_j)=\text{false}$ i.e., ${x^2_j}^\Imc\notin \mathsf{T}^{\Imc_\pmsf}$, which implies ${x^2_j}^\Imc\in \mathsf{F}^{\Imc_\pmsf}$, and for every negative literal $\neg x^2_j$ of $c^2_k$, $\nu(x^2_j)=\text{true}$, i.e., ${x^2_j}^\Imc\in \mathsf{T}^{\Imc_\pmsf}$. Hence, ${c^2_k}^\Imc\in \mathsf{NotSat}^{\Imc_\pmsf}$, so that $\varphi_2^\Imc\in \mathsf{NotSat}^{\Imc_\pmsf}$. It follows by Definition~\ref{def:querysanswer} that $\Kmc\models_\four \mathsf{NotSat}(\varphi_2)$. 

Hence $\Kmc\models_\four\qmbf$.
\end{proof}
\end{document}